\documentclass[12pt]{article}
\newcommand{\be}{\begin{equation}}
\newcommand{\ee}{\end{equation}}
\usepackage{mathtools}
\usepackage{amssymb}
\usepackage{amsthm}
\usepackage{url}
\usepackage{mathrsfs}
\usepackage{enumerate}
\usepackage{tikz}
\usepackage{lscape}
\usepackage[noadjust]{cite}
\usepackage{fancyhdr}
\usepackage[T1]{fontenc}
\usepackage{tikz-cd}
\usepackage{float}

\tikzset{every node/.style={font=\tiny}}

\def\k{{\rm k}}

\let\oldfootnote\footnote
\def\footnote{\ifhmode\unskip\fi\oldfootnote}
\newcommand{\nn}{\nonumber}
\newcommand{\Idm}{\mbox{1\kern -0.5ex I}}
\newcommand{\BR}{\mathbb{R}}
\newcommand{\BC}{\mathbb{C}}

\newcommand{\BK}{\mathbb{K}}
\newcommand{\BZ}{\mathbb{Z}}
\newcommand{\bl}{\bar{\lambda}}
\newcommand{\bL}{\bar{\Lambda}}
\newcommand{\bi}{\bar{i}}
\newcommand{\bbe}{\bar{e}}
\newcommand{\bv}{\mbox{Bil}}
\newcommand{\av}{\mbox{Alt}}
\newcommand{\sv}{\mbox{Sym}}
\newcommand{\qv}{\mbox{Quad}}

\newcommand{\Tprod}[3]{\mathop{#1 #2}\displaylimits_{#3}}
\newcommand{\BAprod}[2]{\mathop{#1 #2}\displaylimits_A} 
\newcommand{\BFprod}[2]{\mathop{#1 #2}\displaylimits_F} 

\newcommand{\w}{\wedge}
\newcommand{\JJF}{\mathop{\rfloor\,}\displaylimits_F}
\usepackage{slashed}

\newtheorem{theorem}{Theorem}[section]
\newtheorem{lemma}[theorem]{Lemma}
\newtheorem{proposition}[theorem]{Proposition}
\newtheorem{cor}[theorem]{Corollary}
\newtheorem{rem}[theorem]{Remark}

\newtheorem{definition}[theorem]{Definition}

\newtheorem*{note*}{Note}
\newtheorem{note}[theorem]{Note}

\def\QED{\hfill \rule{2.5mm}{2.5mm}
\vspace{5pt}\noindent }

\DeclareMathOperator{\sgn}{sgn}
\DeclareMathOperator{\cl}{Cl}
\DeclareMathOperator{\End}{End}
\DeclareMathOperator{\T}{T}
\DeclareMathOperator{\Rad}{Rad}
\DeclareMathOperator{\Id}{Id}
\begin{document}
\title{On the bundle of Clifford algebras over the space of quadratic forms}
\author{Arkadiusz Jadczyk\\
Laboratoire de Physique Th\'{e}orique, Universit\'{e} de Toulouse III 
and\\
Ronin Institute, Montclair, NJ 07043}
\providecommand{\keywords}[1]{\textbf{\textit{Index terms---}} #1}
\keywords{tensor algebra, Clifford algebra, exterior algebra, Chevalley's isomorphism, gauge tranformations }
\maketitle
\begin{abstract}
For each quadratic form $Q\in \qv(V)$ on a vector space over a field $\BK,$ we can define the Clifford algebra $\cl(V,Q)$ as the quotient $\T(V)/I(Q)$ of the tensor algebra $\T(V)$ by the two-sided ideal generated by expressions of the form $x\otimes x-Q(x),\, x\in V.$ In the present paper we consider the whole family $\{\cl(V,Q):\, Q\in \qv(V)\}$ in a geometric way as a $\BZ_2$-graded vector bundle over the base manifold $\qv(V).$ Bilinear forms $F\in\bv(V)$ act on this bundle providing natural bijective linear mappings $\bl_F$ between different Clifford algebras $\cl(V,Q).$ Alternating  (or antisymmetric) forms induce vertical automorphisms, which we propose to consider as `gauge transformations'. We develop here the formalism of N. Bourbaki, which generalizes the well known Chevalley's isomorphism $\cl(V,Q)\rightarrow\End(\bigwedge(V))\rightarrow\bigwedge(V).$ In particular we realize the Clifford algebra twisting gauge transformations induced by antisymmetric bilinear forms as exponentials of contractions with elements of $\bigwedge^2(V^*)$ representing these forms. Throughout all this paper we intentionally avoid using the so far accepted term ``Clifford algebra of a bilinear form'' (known otherwise as ``Quantum Clifford algebra"), which we consider as possibly misleading, as it does not represent any well defined mathematical object. Instead we show explicitly how any given Clifford algebra $\cl(Q)$ can be naturally realized as acting via endomorphisms of any other Clifford algebra $\cl(Q')$ if $Q'=Q+Q_F,\, F\in \bv(V)$ and $Q_F(x)=F(x,x).$ Possible physical meaning of such transformations is also mentioned.
\end{abstract}
\thanks{This paper is dedicated to the memory of Zbigniew Oziewicz. }
\maketitle
\section{Introduction}
This paper was inspired by my studies of previous works dealing with Clifford algebra of multivectors for general, not necessarily symmetric or non-degenerate, bilinear forms. Here I will mention only a few of these works that were the main basis of my own research: Z. Oziewicz \cite{oziewicz97}, R. Ab{\l}amowicz and P. Lounesto \cite{ablamowiczlounesto}, B. Fauser \cite{fauser97}, R. Ab{\l}amowicz and B. Fauser \cite{ablamowiczfauser}. \footnote{More relevant references can be found in the cited papers, and also in a more recent paper by R. Ab{\l}amowicz et al. \cite{ablam2014}.} All these works are based on Chevalley's construction published in 1954 - cf. Ref. \cite[p. 102]{chevalley1996}. Chevalley is using there an explicitly nonsymmetric bilinear form to realize an algebra homomorphism from the Clifford algebra $\cl(V,Q)$ of a quadratic form $Q$ on a finite-dimensional vector space $V$ over an arbitrary field $\BK$ into the algebra of endomorphisms $\End(\bigwedge(V)).$ The specially constructed bilinear form (called $B_0$ in Ref. \cite{chevalley1996}) is explicitly defined to be non-symmetric in order to cover the case of characteristic $2.$ Chevalley himself did not develop this idea any further. But it was subsequently developed a great deal (probably with his participation) in the 1959 algebra textbook ``Alg\`{e}bre, Chapitre 9, Formes Sesquilin\'{e}ares et Formes Quadratiques''  \cite{bourbaki1959} by N. Bourbaki. This particular volume of ``Algebra" by N. Bourbaki is the only one that was not translated into English and, as it seems to me, is totally unknown in the Cliffordian community. In the present paper I am presenting the relevant part from Bourbaki, developing it further so as to make it usable for applications of Clifford algebras, and in particular of `Clifford algebras of multivectors',  to physics. In doing it, I am trying to keep as much generality as possible, therefore, most of the time, not requiring the characteristic to be different from $2$. Much of the constructions developed in the present paper can be done for general modules over rings, but I am deliberately restricting myself to vector spaces, usually not requiring them to be finite-dimensional (unless specifically mentioning otherwise). \footnote{In applications to physics, when discussing Clifford algebras,  we usually need only finite-dimensional vector spaces. A possible exception is the application of Clifford algebras in the discussion of Canonical Anticommutation Relations (CAR) in quantum field theory, but in this case we need functional analysis and $C^*$ algebras rather than pure  algebraic constructions with infinite-dimensional vector spaces equipped with Hamel bases.}

One important distinctive feature of the Bourbaki approach consists of the use of operators $e_x,\,x\in V$ of left multiplication, and $\BZ_2$-graded derivations $i_f,\, f\in V^*$ already at the level of the tensor algebra $\T(V),$ where $i_f$ is defined recursively through $i_f(\mathbf{1})=0,$ $i_f(x\otimes u)=f(x)u- x\otimes i_f(u)$ for $x\in V,\,u\in \T(V).$ \footnote{Cf.  \cite[\S 9.2]{bourbaki1959}. The antiderivation $i_f$ acting on the tensor algebra can be also found in Ref. \cite[Ch. 2.2.9]{meinrenken2013}.} Usually, $\BZ_2$ derivations (antiderivations) are introduced only at the level of the exterior or Clifford algebra, that is after passing to the quotient $\T(V)/I(Q),$ where $I(Q)$ is the two-sided ideal generated by the expressions $x\otimes x -Q(x)$, $x\in V,$ where $Q$ is a quadratic form ($Q=0$ for the exterior  algebra). The point is that $\T(V)$ is naturally $\BZ$-graded, and the reduction to $\BZ_2$-gradation  may seem to be somewhat artificial. And yet the Bourbaki original approach allows for the construction of natural linear mappings $\lambda_F:\T(V)\rightarrow \T(V)$ that map every two-sided ideal $I(Q)$ to $I(Q-Q_F)$, where, for any bilinear form $F\in \bv(V)$,  $Q_F$ is the quadratic form $Q_F(x)=F(x,x)$. I consider Bourbaki's mapping $\lambda_F:\T(V)\rightarrow \T(V)$  to be the most important tool for developing the subject of deformations of Clifford algebra products.\footnote{While it is possible to consider all Clifford algebras as deformations of one algebra, the exterior algebra, invoking the tensor algebra allows us to have a 'bird view' of the whole structure: all Clifford algebras, including the exterior algebra, have one `mother', and this mother is the tensor algebra.} For $x\in V$ and $F\in \bv(V),$\, one first defines $i_x^F$ as $i_f$ for the linear form $f$ defined as $f(y)=F(x,y).$ Then $\lambda_F:\T(V)\rightarrow \T(V)$ is defined recursively as $\lambda_F(\mathbf{1})=\mathbf{1}$ on $\T_0(V),$ and $\lambda_F(x\otimes u)=x\otimes\lambda_F(u)+i_x^F(\lambda_F(u)).$   Using $\lambda_F$  way we can ``travel'' between different Clifford algebras already at the level of $\T(V)$. In fact, one can define the analogue of the Chevalley map, I call it $\Lambda_F: \T(V)\mapsto\End(\T(V)),$ already at the level of the tensor algebra. Then $\lambda_F(u)=\Lambda_F(u)(\mathbf{1})$, $\mathbf{1}\in \T(V),$ in a full analogy to Chevalley's construction of the representation of a Clifford algebra by endomorphisms of the exterior algebra, and then realizing the Clifford product within the exterior algebra. All these operators $e_x,i_f, i_x^F,\lambda_F,\Lambda_F$ pass to the quotient and define corresponding operators within and between Clifford algebras.\footnote{The operators $\bi_x^F,\bar{e}_x,\bL_x^F,\bl_x^F$ can be also found in Ref. \cite[Ch. 5.7]{Northcott2009}, denoted there as $\Delta_x,L_x,\Lambda_x,\Omega(x)$ resp. } In this paper I am following Bourbaki's convention and denote the descendants of these operators, acting at the level of Clifford algebras,  with a bar $\bar{e}_x,\bi_f,\bi_x^F,\bl_F,\bL_F$.\footnote{Bourbaki's text is not exactly  consistent here, because $e_x$ is denoted the same way for $\T(V)$, where $e_x(u)=x\otimes u$, and for exterior or Clifford algebra, where $e_x(u)=x\wedge u$ or $e_x(u)=xu.$ } It is because of its importance that the explicit form of $\lambda_F$ is developed in Proposition \ref{prop:ln}, and then $\lambda_F$ is shown to be an exponential of an analogue of a two-fermion annihilation operator in Sec. \ref{sec:lexp}.\footnote{In characteristic $\neq 0,$ instead of the usual exponential power series, an exponential series with  divided powers is given.} Later on, in Sec. \ref{sec:aaf}, at the level of Clifford (or exterior) algebras, and for $F$ restricted to be an alternating  form $A\in \av(V)$, $\bl_A$ is shown to be the exponential of $\bi_{A^*},$ where $A^*\in\bigwedge^2 (V^*)$ represents $A$ via the duality $\langle A^*,x\wedge y\rangle=A(x,y)$, cf. Eq. (\ref{eq:blexp}). This allows us to write the formula for a deformation of a Clifford product by an alternating  form as (cf. Eq. (\ref{eq:xAv})):
\be \BAprod{x}{u}=e^{\bi_{A^*}}\left(x\,e^{-\bi_{A^*}}(u)\right),\, x\in V,\, u\in \cl(V,Q),\nn\ee
which I propose here instead of the ``Wick isomorphism'' based on the exponential of an element of $\bigwedge^2(V) $ suggested in Ref. \cite[Sec. 3]{ablamowiczfauser}.\footnote{Cf. Note \ref{n:wick}} \\
More generally, and in an arbitrary characteristic, we can use the mapping $\bl_F$ to deform a given Clifford multiplication by an arbitrary bilinear form $F$:
\be \BFprod{x}{u}=\bl_F\left(x\bl_{F}^{-1}(u)\right),\, x\in V,\, u\in\cl(V,Q),\ee
which is usually written on the background of the exterior algebra (the case of $Q=0$) as
\be \BFprod{x}{u}=\bl_F\left(x\wedge\bl_{F}^{-1}(u)\right),\, x\in V,\, u\in\bigwedge(V).\ee
Since $\bl_F$ has the exponential property $\bl_{F+G}=\bl_F\circ\bl_G,$ when $F=g+A$\footnote{Here we use the notation of Ref. \cite{ablamowiczfauser}.}, with $g$ symmetric and $A$ antisymmetric\footnote{In characteristic $\neq 2$ antisymmetric and alternating  bilinear forms coincide.}, we get
\begin{eqnarray*} \BFprod{x}{u}&=&\bl_F(x\wedge\bl_{F}^{-1}(u))=\bl_A\left(\bl_g(x\wedge\bl_g(\bl_A(u))\right)\\
&=&\bl_A(x\cdot_g\,\bl_{A}^{-1}(u))=x\wedge u+\bi_x^g(u)+\bi_x^A(u).\end{eqnarray*}
The product defined by $x\cdot_F\,u$ is a twisted (with respect to $x\cdot_g\, u$) representation of the Clifford algebra $\cl(V,Q_g)$ on the exterior algebra $\bigwedge(V).$ In Ref. \cite{ablamowiczlounesto} it is denoted $\cl(V,F)$ and referred to  as  ``Quantum Clifford algebra''. Yet it is just a different but equivalent realization of the same Clifford algebra $\cl(Q_g)$, and therefore  it has the same irreducible representations (ungraded and $\BZ_2$ graded) - cf. Sec. \ref{sec:reps}. In particular the twisted representation of the Clifford algebra $\cl(V,Q_g)$ for $g$ of signature $(2,2)$, twisted by a nontrivial antisymmetric form $A$, on its $8$-dimensional left ideal,  suspected to be irreducible in Ref. \cite{ablamowiczlounesto}, turns out to be reducible.\footnote{I am indebted to R. Ab{\l}amowicz for an informative and extensive discussion of this subject.}\\

It is instructive to think of the family of Clifford algebras $\{\cl(V,Q):Q\in \qv(V)\}$ as a $\BZ_2$-graded vector bundle $\cl(V)$ over the space $\qv(V)$ of all quadratic forms on $V$. While each fiber $\cl(V,Q)$ carries an algebra structure, we will assign an important role to the maps $\bl_F,\, F\in \bv$, that are linear bijections preserving the $\BZ_2$-gradation, but not algebra homomorphisms. The additive group of the vector space $\bv(V)$ acts on the base space via $F:Q\mapsto Q+Q_F$, and it acts on the bundle space $\cl(V)$ via the maps $\lambda_F$, mapping fibers onto fibers. The elements from the subgroup $\av(V)\subset \bv(V)$ do not move the points of the base and transform each fiber $\cl(V,Q)$ into itself, deforming the product within each algebra. In the theory of fiber bundles transformations of the bundle space that map bijectively fibers into fibers, and therefore induce a transformation of the base manifold, are called bundle automorphisms. Bundle automorphisms that do not move base points, so called ``vertical automorphisms", are often referred to as {\it gauge transformations}, more precisely, as {\it global active
gauge transformations\,} - cf. \cite[Ch. 3.2]{bleecker2005} and \cite[Appendix H]{percacci1986}. Instead of considering the usual Dirac equation, where the Clifford algebra acts on an irreducible graded ``spin module'' one considers the so called Ivanenko-Landau-K\"{a}hler equation \cite{obukhov1985,obukhov1994}, where spinors are represented by differential forms - elements of the exterior bundle of the cotangent space. The idea of representing spinors as differential forms has its clear physical justification that was stated explicitly by P. Lounesto in the following form (cf. Ref. \cite[p. 145]{lounesto2001}):
\begin{quotation}
``\dots However the physical justification of the theory of spin manifolds could be questioned on the following basis: why should we need to know the global properties of the universe if we want to explore the local properties of a single electron?''
\end{quotation}
Lounesto was at the same time  well aware of the fact that choosing a global minimal graded ideal of the Clifford algebra bundle over space-time is questionable (cf. Ref. \cite[p. 145]{lounesto2001}):
\begin{quotation}
``\dots In curved manifolds it is more appropriate to use abstract representation modules as spinor spaces and not minimal left ideals [not even subalgebras] of Clifford algebra. The injection ties these spaces together in a manner that singles out special directions in $\BR^{1,3}.$''
\end{quotation}
The only acceptable solution to these problems is by using the whole exterior algebra of differential forms, as in Ivanenko-Landau-K\"{a}hler equation on a curved spacetime manifold $M$, where gravitation is described via teleparalellism, along the lines outlined in Ref. \cite{obukhov2017}. There are at least two ways in which the deformations of the Clifford algebra action by antisymmetric forms can enter into field equations for matter fields in a gravitational background. The first way is by twisting the action of the Dirac gamma matrices $\gamma^\mu$ on the exterior algebra: $\gamma^\mu\rightarrow e^{\bi_{A^*}}\gamma^\mu e^{-\bi_{A^*}}.$ If $A$ is an antisymmetric form that does not depend on space-time coordinates $x^\mu$, then this can be compensated by redefining the matter field $\Phi(x)\in \bigwedge(T_x^*(M))$ via $\Phi\rightarrow e^{\bi_{A^*}}\Phi.$ When $A$ depends on space-time coordinates, we then get extra terms in the Ivanenko-Landau-K\"{a}hler equation resulting from the derivatives of $A(x)$. We can then treat $A$ as a non-minimal interaction in matter equations, or, alternatively, we can consider $\Phi(x)\rightarrow e^{\bi_{A^*(x)}}\Phi(x)$ as a local gauge transformation, in which case we should expect appearance of a certain non-closed $3$-form $\Psi$ in the equation, that transforms under these gauge transformations as $\Psi\rightarrow \Psi+dA.$ These ideas about a possible physical significance of twisted Clifford algebra products are in a need of further investigation, and they complement those already suggested in Refs. \cite{fauser97,ablamowiczfauser}.

\subsection{General setup}
A reasonably general formulation of the theory of Clifford algebras starts with the definition of the Clifford algebra of a module over a commutative ring, equipped with a quadratic form,  see, e.g., Ref. \cite[p. 139]{bourbaki1959}. Here we will not aim for such a generality, and we restrict ourselves to Clifford algebras over a vector space $V$ over a field $\BK,$ equipped with a quadratic form $Q.$\footnote{A vector space is a (projective) module over a (commutative) field.} Unless explicitly stated we will not assume $V$ to be finite dimensional.  Initially, following Ref. \cite{bourbaki1959} we will not even demand for $\BK$ to be of characteristic $\neq 2.$ In electrodynamics, gauge transformations are implemented via exponentials of the type $\exp(i\phi),$ where $\phi$ is a real scalar. These exponentials act on complex wave functions changing their phase. In our case we have exponentials of the type $\exp(\bi_{A^*})$, where $A$ is an antisymmetric bilinear form. These exponentials twist algebra products in Clifford algebras $\cl(V,Q)$  acting on the exterior algebra $\bigwedge(V).$ This brings us the idea of representing physical fields, including those corresponding to spin $1/2,$ as multivectors, that is, the elements of the exterior algebra, as it was extensively discussed in Ref. \cite{graf1978};
while it is instructive to see how the case of characteristic 2 is being dealt with in Ref. \cite{bourbaki1959}, our main interest will be the case of $\BK$ being the field of real numbers $\BR$ or complex numbers $\BC.$

In linear algebra, one shows that every vector space has a basis, possibly infinite, of linearly independent vectors (Hamel basis). Moreover, every two bases have the same cardinal number called the dimension of the vector space - cf. e.g. \cite[p. 103]{chevalley1956}. Every system of linearly independent vectors can be extended to a basis.
 \label{sec:scl}

\subsection{Tensor algebra of $V$}
Let $V$ be a vector space over $\BK.$ An algebra $\T$ is called a tensor algebra over $V$ (or ``of $V$") if it satisfies the following {\em universal property}\footnote{In all standard textbooks, see e.g. \cite{Bourbaki1998,chevalley1996,Northcott2009,greub1978}, the above characterisation of the tensor algebra of a module is always completed by an explicit construction. We also notice that the property of $\T(V)$ to be `generated by $V$' in (i) is superfluous (and even `forbidden'), as it is a consequence of the `uniqueness' that is the part of the universal property (ii).}
\begin{enumerate}
\item[{\rm (i)}] $\T$ is an associative algebra with unit containing $V$ as a linear subspace, and is generated by $V,$
\item[{\rm (ii)}] Every linear mapping $\phi$ of $V$ into an associative algebra $A$ with unit over $\BK$, can be uniquely extended to an algebra homomorphism $\theta$ of $\T$ into $A$:
\end{enumerate}
Denoting by $\iota$ the embedding of $V$ into $\T$ mentioned in (i) the universal property expressed in (ii) reads:
\be \phi=\theta\circ\iota.\nn\ee  Let $\T(V)$ be the tensor algebra of $V$. The multiplication in $\T(V)$ is denoted by the symbol $\otimes$. If $\{e_i\}_{i\in I(V)}$ is a basis in $V$, then $\mathbf{1}\in \T^0(V)\subset \T(V)$ together with $e_{i_1}\otimes\cdots\otimes e_{i_p}$, $(p=1,2,\ldots )$ form a basis in $\T(V)$.
The tensor algebra $\T(V)$ of $V$ is $\BZ$-graded. We have
 \be \T(V)=\bigoplus\displaylimits_{p=0}^{\infty}\T^p(V),\nn\ee
 where
 \be \T^p(V)=V^{\otimes p}=\underbrace{V\otimes \cdots \otimes V}_\text{$p$ \textup{factors}}\nn\ee
 is the subspace spanned by $e_{i_1}\otimes\cdots\otimes e_{i_p}$.
It is understood here that $\T^0(V)=\BK$ and $\T^1(V)=V.$ The fact that $\T(V)$ is a graded algebra means that for any $x\in \T^p(V),y\in \T^q(V)$ the product $xy$ is in $\T^{p+q}(V)$ for all $p,q=0,1,\dotsi $.

By using the universal property, one defines the main involution $\alpha$ and the main anti-involution $\tau$ of $\T(V).$\footnote{The mapping $\alpha$ is an involutive algebra homomorphism, while $\tau$ is an involutive algebra anti-homomorphism.} Explicitly, on homogeneous elements, we have:
\be \alpha (x_1\otimes\cdots\otimes x_p)=(-1)^p x_1\otimes \cdots \otimes x_p,\label{eq:alpha}\ee
\be \tau(x_1 \otimes\cdots\otimes x_p)=x_p\otimes \cdots\otimes x_1,\label{eq:tau}\ee
for $x_1,\ldots,x_p\in V.$

The algebra $\T(V)$ is $\BZ_2$-graded into even and odd parts.   The main involution $\alpha$ respects this $\BZ_2$-gradation, For even $u\in \T(V)$ we have $\alpha(u)=u$, and for $u$ odd we have  $\alpha(u)=-u.$
\subsubsection{Operators $e_x$, $i_f$, $i_x^F$.}\label{sec:ops}
For any vector space $W$ we denote by $\End(W)$ the algebra of all endomorphisms (linear maps) of $W.$  For $x\in V$, we denote by $e_x$ the linear operator $e_x\in \End(\T(V)),$ $\T^p(V)\rightarrow \T^{p+1}(V),$ of left multiplication by $x$:
\be e_x:u\mapsto e_x(u)=x\otimes u,\, u\in \T(V).\nn\ee
We denote by $V^*$ the dual vector space, that is the space of all linear functions from $V$ to $\BK.$ The following Proposition associates to each $f\in V^*$ an antiderivation of the graded algebra $\T(V)$.\footnote{C.f. also \cite[Lemma 3.2, p. 43]{chevalley1996},\cite[Lemma 1, p. 141]{Bourbaki2006}. Usually $i_f$ is defined on the exterior algebra rather than on the tensor algebra. Notice however that $i_f$ defined here is not the same as $i_f$ defined in Ref. \cite[A III.161-165]{Bourbaki1998} and\cite[A.15.2.1, p. 367]{dieu3}, where $i_f$ on $\T(V)$ is defined as the transpose of the operator  $e_f$ acting on $\T(V)^*$, and where $i_f(u)$ is also written as $f\rfloor\, u$ and called a {\it a contraction in the direction of $f$. }}
\begin{proposition}\label{prop:if}
Let $f$ be an element of $V^*.$ There exists a unique linear mapping $i_f$ from $\T(V)$ to $\T(V)$ such that
\begin{enumerate}
\item We have \be i_f(\mathbf{1})=0,\label{eq:if1}\ee
\item For all $x\in V$ we have
\be e_x\circ i_f+i_f\circ e_x=f(x)\mathbf{1}\label{eq:if2a}\ee
\end{enumerate}
The map $f\mapsto i_f$ from $V^*$ to linear transformations on $\T(V)$ is linear. We have
\begin{enumerate}
\item[{\rm (i)}] For all $x\in V\subset \T(V)$, $u\in \T(V),$ \be i_f(x\otimes u)=f(x)u-x\otimes i_f(u),\label{eq:ift}\ee
\item[{\rm (ii)}] $i_f(\T^p(V))\subset \T^{p-1}(V),$
\item[{\rm (iii)}] $i_f^2=0$,
\item[{\rm (iv)}] $i_fi_g+i_gi_f=0,\,\mbox{for all }f,g\in V^*$.
\end{enumerate}
For $x_1,\ldots,x_p\in V$ we have
\be i_f(x_1\otimes\cdots\otimes x_p)=\sum_{i=1}^p (-1)^{i-1}f(x_i)\,x_1\otimes\cdots\otimes \hat{x}_i\otimes\cdots\otimes x_p.\label{eq:ifx1p}\ee
\label{prop:if0}\end{proposition}
\begin{proof}
The formula (\ref{eq:ift}) in (i) is an explicit expression of Eq. (\ref{eq:if2a}). The uniqueness of $i_f$ follows from the fact that Eq. (\ref{eq:if1}) determines $i_f$ on $\T^0(V)$, and
Eq. (\ref{eq:ift}) determines $i_f$ on $\T^p(V)$ once it is given on $\T^{p-1}(V)$. To prove the existence we notice that for $p\geq 1$ the map
\be (x_1,\ldots,x_p)\mapsto \sum_{i=1}^p (-1)^{i-1}f(x_i)\,x_1\otimes\cdots\otimes \hat{x}_i\otimes\cdots\otimes x_p\nn\ee
from $V\times\cdots\times V$ to $\T(V)$ is $p$-linear, and therefore, by the universal property of tensor products, extends to a unique linear map, denoted here by $i^{(p)}_f,$ from $\T^p(V)$ to $\T^{p-1}(V).$ The mappings $i^{(p)}_f$, together with Eq.(\ref{eq:if1}) determine $i_f$ on $\T(V)$, so that Eq. (\ref{eq:ifx1p}) holds, and then Eq. (\ref{eq:ift}) follows by  linearity. To prove (iii) we apply $i_f$ to both sides of Eq. (\ref{eq:ift})  to obtain
\begin{eqnarray*} i_f^2(x\otimes u)&=&f(x)i_f(u)-f(x)i_f(u)+x\otimes i_f^2(u)\\
&=&x\otimes i_f^2(u).\end{eqnarray*} Since $i_f(\mathbf{1})=0$, thus $i_f^2(\mathbf{1})=0,$ by recurrence we get that $i_f^2(u)=0$ for all $u\in \T(V).$ Then (iv) follows by noticing that $i_{f+g}^2=0$ for all $f,g\in V^*.$
\end{proof}
\begin{definition}[$i_x^F$]\label{def:ixf}
Let $F$ be a bilinear form on $V.$ Then every $x\in V$ determines a linear form $f_x$ on $V$ defined as $f_x(y)=F(x,y).$ {\em We will denote by $i_x^F$ the antiderivation $i_{f_x}$  described in Proposition \ref{prop:if0}}. In particular we have:
\begin{enumerate}
\item[{\rm (i)}] $i_x^F(\mathbf{1})=0$,
\item[{\rm (ii)}] For all $y\in V$, $w\in \T(V)$ we have \be i_x^F(y\otimes w)=F(x,y)w-y\otimes i_x^F(w),\label{eq:iFx0}\ee
\end{enumerate}
\end{definition}
\begin{proposition}\label{prop:ixf1}
With the same  notation as in the Definition \ref{def:ixf}, for $y_1,\ldots,y_n$ in $\T(V)$ we have
\be i_x^F(y_1\otimes\cdots\otimes y_n)=\sum_{j=1}^n (-1)^{j-1}F(x,y_j)y_1\otimes\cdots\otimes\hat{y}_j\otimes\cdots\otimes y_n,\nn\ee
where $\hat{y}_j$ means that this factor is omitted in the product.
\end{proposition}
\begin{proof}
The proof follows immediately from the definition.
\end{proof}
\subsection{Bourbaki's map $\lambda_F$}Bourbaki's book on sesquilinear and quadratic forms \cite{bourbaki1959} is one of the very few in the N. Bourbaki ``\'{E}l\'{e}ments de math\'{e}matique'' series that has never been translated into English.\footnote{It has been translated into Russian though \cite{bourbaki_ru}.} That is probably one of the main reasons why the idea of constructing a family of maps $\lambda_F$, parametrized by bilinear forms $F$ and acting as bundle automorphisms in the $\BZ_2$ graded vector bundle of Clifford algebras over the space of bilinear forms,  did not receive the deserved attention in the Clifford algebra community.\footnote{Ref. \cite[pp. 30-31]{crumeyrolle} seems to be an exception. The author there considers Bourbaki's mapping $\lambda_F$, though without quoting Bourbaki, for unexplained reason he calls it `antiderivation', and fails to distinguish between its actions on the tensor, exterior and Clifford algebras.} Here we will introduce Bourbaki's operators $\Lambda_F$ by modifying Chevalley's method as described in \cite[Ch. 2.1]{chevalley1996}, \cite[Ch. 5.14]{chevalley1956}, \cite[Ch. 5.7]{Northcott2009}. \footnote{Here we use the tensor algebra, instead of the exterior algebra as it is done in the quoted references.}\textsuperscript{,} \footnote{In Ref. \cite{Helmstetter2008} the authors discuss deformations of Clifford algebras, but they do it in their own way, difficult to follow for the present author.}
\begin{definition}
Let $F$ be a bilinear form on $V$. For each $x\in V$ let $\Lambda_F(x)\in \End(\T(V))$ be defined as
\be \Lambda_F(x)=e_x+i_x^F.\nn\ee
The map $x\mapsto \Lambda_F(x)$ is linear, therefore, by the universal property of $\T(V)$, it extends to a unique algebra homomorphism, denoted $\Lambda_F$, from $\T(V)$ to $\End(\T(V))$
\be \Lambda_F:\T(V)\ni u \mapsto \Lambda_F(u)\in \End(\T(V)).\nn\ee
\end{definition}
In particular,
\be \Lambda_F(\mathbf{1})=\Id_{\T(V)}.\label{eq:L1}\ee
For $u,v\in \T(V)$, we have
 \be \Lambda_F(u\otimes v)=\Lambda_F(u)\,\Lambda_F(v),\label{eq:L1a}\ee
and, in particular,  for $x\in V$, $u\in \T(V)$  we have
 \be \Lambda_F(x\otimes u)=(e_x+i_x^F)\Lambda_F(u).\label{eq:Lfxu}\ee
Finally, for $x_1,\ldots,x_p\in V$, we have
\be \Lambda_F(x_1\otimes\cdots\otimes x_p)=(e_{x_1}+i_{x_1}^F)\cdots (e_{x_p}+i_{x_p}^F).\label{eq:L1p}\ee
\begin{proposition}
For $u\in \T^p(V)$ and $f\in V^*$ we have
\be i_f\circ \Lambda_F(u)=\Lambda_F(i_f(u))+(-1)^p\Lambda_F(u)\circ i_f.\label{eq:iFu}\ee
\end{proposition}
\begin{proof}
We prove by induction. Since $\Lambda_F(\mathbf{1})=\Id$, the statement evidently holds for $u=\mathbf{1}.$ Let us suppose that the statement holds for all $u\in \T^p(V)$, $p\geq 0$. Then with $u\in \T^p(V)$, $x\in V$ we have
\begin{equation*} \begin{split}
i_f\circ\Lambda_F(x\otimes u)&=i_f\circ(e_x+i_x^F)\,\Lambda_F(u)=(i_f e_x +i_fi_x^F)\Lambda_F(u)\\
&=(f(x)-e_xi_f-i_x^f i_f)\Lambda_F(u)\\
&=f(x)\Lambda_F(u)-(e_x+i_x^F)i_f\circ \Lambda_F(u)\\
&=f(x)\Lambda_F(u)-(e_x+i_x^F)\big(\Lambda_F(i_f(u))+(-1)^p\Lambda_F(u)\circ i_f\big)\\
&=\Lambda_F(f(x)u)-\Lambda_F(x\otimes i_f(u))+(-1)^{p+1}\Lambda_F(x\otimes u)\circ i_f\\
&=\Lambda_F(i_f(x\otimes u))+(-1)^{p+1}\Lambda_F(x\otimes u)\circ i_f.
\end{split}\end{equation*}
Therefore the statement holds also on $\T^{p+1}(V).$
\end{proof}
\begin{lemma}
Let $F,G$ be two bilinear forms on $V$. For $u,v\in \T(V)$ we have
\be \Lambda_F(\Lambda_G(u)(v))(w)=\Lambda_{F+G}(u)(\Lambda_F(v)(w))\label{eq:LFG}\ee
for all $w\in \T(\Rad(G))\subseteq \T(V)$, where $\Rad(G)=\{w\in V: G(v,w)=0,\,\forall v\in V\}.$
\end{lemma}
\begin{proof}
The proof is by induction. From Eq.(\ref{eq:L1}) for $u=\mathbf{1}$ we get, on the left hand side
$\Lambda_F(\Lambda_G(\mathbf{1})(v))(w)=\Lambda_F(v)(w),$ and the same result on the right hand side $\Lambda_{F+G}(\mathbf{1})(\Lambda_F((v)(w))=\Lambda_F(v)(w),$ for all $w\in V.$
Let us assume that the statement holds for $u\in \T^p(V).$
We will show that then it also holds for $u\in \T^{p+1}(V).$ Indeed, for $u\in \T^p(V),$ $v\in \T^q(V),$ and $x\in V$
we have:
\be
\begin{split}
&\Lambda_F(\Lambda_G(x\otimes u)(v))(w)=\Lambda_F((e_x+i_x^G)\,\Lambda_G(u)(v))(w)\nn\\
&=\Lambda_F(e_x\,\Lambda_G(u)(v))(w)+\Lambda_F(i_x^G\Lambda_G(u)(v))(w)\nn\\
&=(e_x+i_x^F)\,\Lambda_F\big(\Lambda_G(u)(v)\big)(w)\\
&+i_x^G\,\Lambda_F\big(\Lambda_G(u)(v)\big)(w)-(-1)^q\Lambda_F\big(\Lambda_G(u)(v)\big)(i_x^Gw)\nn\\
&=(e_x+i_x^{F+G})\,\Lambda_F\big(\Lambda_G(u)(v)\big)(w)-(-1)^q\Lambda_F(\Lambda_G(u)(v))(i_x^Gw).\nn
\end{split}\ee
By the induction hypothesis for $w\in \T(\Rad(G))$ we have $\Lambda_F\big(\Lambda_G(u)(v)\big)(w)=\Lambda_{F+G}(u)\big(\Lambda_F(v)(w)\big)$, therefore
\be \begin{split}
&\Lambda_F(\Lambda_G(x\otimes u)(v))(w)\\
&=(e_x+i_x^{F+G})\,\Lambda_{F+G}(x\otimes u)(v))(w)-(-1)^q\Lambda_F\big(\Lambda_G(u)(v)\big)(i_x^Gw)\\
&=\Lambda_{F+G}(x\otimes u)(\Lambda_F(v)(w))-(-1)^q\Lambda_F\big(\Lambda_G(u)(v)\big)(i_x^Gw).
\nn\end{split}\ee
Now, if $w\in \T(\Rad(G))$, the last term vanishes, which proves the lemma.
\end{proof}
We notice that if $G$ is nondegenerate, then $\Rad(G)=\{0\}$, and $$\T(\Rad(G))=\T^0(V)=\BK \mathbf{1}.$$
\begin{definition}
Let $F$ be a bilinear form on $V$. We define $\lambda_F\in \End(\T(V))$ as
\be \lambda_F(u)=\Lambda_F(u)(\mathbf{1}), \label{eq:lf1}\ee
where $\mathbf{1}$ denotes $\mathbf{1}\in \BK=\T^0(V).$
\end{definition}
\begin{proposition}\label{prop:7}
The map $\lambda_F$ defined above has the following properties
\begin{enumerate}
\item[{\rm (i)}]
\be \lambda_F(\mathbf{1})=\mathbf{1},\quad \lambda_F(x)=x,\,(x\in V),\label{eq:l1}\ee
\item[{\rm (ii)}] \be\lambda_F(x\otimes u)=i^F_x(\lambda_F(u))+x\otimes\lambda_F(u),\, (x\in V)\label{eq:lb},\ee
\item[{\rm (iii)}] For all $f\in V^*$ we have
\be \lambda_F\circ i_f=i_f\circ\lambda_F.\label{eq:iii}\ee
\item[{\rm (iv)}] If $F$ and $G$ are two bilinear forms on $V,$ then \be\lambda_F\circ\lambda_G=\lambda_{F+G}.\label{eq:lbc0}\ee
\item[{\rm (v)}] For $F=0$, $\lambda_F$ is the identity mapping:
\be \lambda_0 =\Id_{\T(V)}.\nn\ee
\item[{\rm (vi)}] For every bilinear form $F,$ the linear mapping $\lambda_F:\T(M)\rightarrow \T(M)$ is a bijection.
\item[{\rm (vii)}] The map $\lambda_F$ preserves the parity (even-odd), i.e. with $\alpha$ defined as in Eq. (\ref{eq:alpha}) for $x_1,\ldots ,x_p\in V$ we have
\be \alpha (\lambda_F(x_1\otimes\cdots\otimes x_p))=(-1)^p\,\lambda_F(x_1\otimes\cdots\otimes x_p).\nn\ee
\end{enumerate}
\end{proposition}
\begin{proof}
We have $$\lambda_F(\mathbf{1})=\Lambda_F(\mathbf{1})(\mathbf{1})=\Id_{\T(V)}(\mathbf{1})=\mathbf{1},$$ and $\lambda_F(x)=(e_x+i_x^F)(\mathbf{1})=x,$ therefore (i). From Eq. (\ref{eq:iFu}) we have
$$\lambda_F(x\otimes u)=\Lambda_F(x\otimes u)(\mathbf{1})=(e_x+i_x^F)\Lambda_F(u)(\mathbf{1})=(e_x+i_x^F)\lambda_F(u).$$ Therefore (ii) holds. Applying Eq. (\ref{eq:iFu}) to $\mathbf{1}\in \T(V),$ and using $i_f(\mathbf{1})=0,$  we get (iii). To prove (iv) we set $v=w=\mathbf{1}$ in Eq. (\ref{eq:LFG}). On the left hand side we get $\Lambda_F(\Lambda_G(u))$, while on the right hand side we notice that $w$ is now in $\T(\Rad(G))$ and we use $\Lambda_F(\mathbf{1})(\mathbf{1})=\mathbf{1}$ to get $\Lambda_{F+G}(u)$, as required. To show (v) notice that  $i_x^0=0,$ and therefore from Eq. (\ref{eq:L1p}) we have $\Lambda_0(x_1\otimes\cdots\otimes  x_p)=e_{x_1}\cdots e_{x_p}.$ It follows that $\lambda_0(x_1\otimes\cdots\otimes x_p)=e_{x_1}\cdots e_{x_p}1=x_1\otimes\cdots\otimes x_p$, and therefore $\lambda_0=\Id_{\T(V)}$. Now, using (iv) and (v), we get $\lambda_F\circ \lambda_{-F}=\lambda_{-F}\circ \lambda_{F}=\lambda_0=\Id_{\T(V)},$
therefore $\lambda_{-F}=(\lambda_F)^{-1}$, and so $\lambda_F$ is invertible. Thus (vi) holds. To prove (vii) notice that $\Lambda_{F,x}=e_x+i_x^F$ changes the parity:
\be \Lambda_{F,x}\circ \alpha=-\alpha\circ \Lambda_{F,x}.\nn\ee
Therefore, using Eq. (\ref{eq:L1p})we obtain
\be \Lambda_F(x_1\otimes\cdots\otimes x_p)\circ \alpha=(-1)^p\alpha\circ \Lambda_F(x_1\otimes\cdots\otimes x_p).\nn\ee
Applying both sides of the above equation to $\mathbf{1}\in \T(V)$ we get (vii).
\end{proof}
\begin{definition}
For $u\in \T(V)=\bigoplus_{p=0}^\infty \T^p(V)$ we will denote by $(u)_p$ the component of $u$ in $\T^p(V).$
\end{definition}
From the definition of the algebraic direct sum, it follows that for each $u\in \T(V)$ we have that $(u)_p=0$, except for a finite number of $p.$
\begin{lemma}\label{lem:par}
For $x_1,\ldots,x_p\in V$ we have that \be (\lambda_F(x_1\otimes\cdots\otimes x_p))_{p+k}=0,\, \mbox{for all } k>0.\nn\ee
\end{lemma}
\begin{proof}
Since $\lambda_F(\mathbf{1})=\mathbf{1}$ and $\lambda_F(x)=x,$ the statement in the lemma is true for $p=0$ and $p=1.$ Let us assume that it holds for a given $p$, and we will show that then it also holds for $p+1.$ Indeed, with $x,x_1,\ldots,x_p\in V$ we have
\be \lambda_F(x\otimes x_1\otimes\cdots\otimes x_p)=x\otimes \lambda_F(x_1\otimes\cdots\otimes x_p)+i_x^F (\lambda_F(x_1\otimes\cdots\otimes x_p)).\label{eq:lfhp}\ee From our induction assumption $\lambda_F(x\otimes x_1\otimes\cdots\otimes x_p)$ has no nonzero components of order higher than $p.$ Thus the first term on the right hand side of Eq. (\ref{eq:lfhp}) has no nonzero components of order higher than $p+1$, while the second term does not have nonzero components of order higher than $p-1$. Thus the assertion in the lemma holds.
\end{proof}
\begin{definition}For integers $p,k,$ with $p\geq 2,$ $2\leq 2k\leq p,$ let $\mathcal{P}_{p,k}$ denote the set of permutations $\pi$ of $\{1,...,p\}$ such that
\begin{enumerate}
\item $\pi(1)<\pi(2),\, \pi(3)<\pi(4),\,\ldots,\,\pi(2k-1)<\pi(2k),$
\item $\pi(1)<\pi(3)<\cdots<\pi(2k-1)$ \quad (if $k\geq 2),$
\item $\pi(j)$ is increasing for $j=2k+1,\ldots,p.$
\end{enumerate}
\label{def:perm}\end{definition}
\begin{note}
In the proof of Proposition \ref{prop:ln} an expression like $\mathcal{P}_{p,k}$ may also be used for permutations $\pi$ of a set of integers $\{i_1,i_2,\ldots,i_p\}$ that may be different from $\{1,2,\ldots,p\}$. If $i_1<i_2<\cdots <i_p,$ Definition \ref{def:perm}
remains meaningful when $1,2,\ldots,p$ have been replaced respectively with $i_1,i_2,\dots,i_p.$
\label{note:sigma}\end{note}
\begin{proposition}
Let $F$ be a bilinear form on $V,$ and assume $p\geq 2$, $x_1,\ldots,x_p\in V$. Then
\be
\begin{split}
&\lambda_F(x_1\otimes\cdots\otimes x_p)=x_1\otimes\cdots\otimes x_p\\
&+\sum_{{2\leq 2k\leq p}}\,\sum_{\pi\in\mathcal{P}_{p,k}}\sgn(\pi) F(x_{\pi(1)},x_{\pi(2)})F(x_{\pi(3)},x_{\pi(4)})\cdots\\
&\cdots F(x_{\pi(2k-1)},x_{\pi(2k)})\,x_{\pi(2k+1)}\otimes\cdots\otimes x_{\pi(p)}.
\end{split}
\label{eq:ls}\ee
\label{prop:ln}\end{proposition}
\begin{note*}
For $p=2$ we have
\be \lambda_F(x_1\otimes x_2)=x_1\otimes x_2+F(x_1,x_2).\nn\ee
For $p=3$:
\be \lambda_F(x_1\otimes x_2\otimes x_3)=x_1\otimes x_2\otimes x_3+F(x_1,x_2)x_3-F(x_1,x_3)x_2+F(x_2,x_3)x_1.\nn\ee
In general, when $p$ is even and $2k=p$, the tensor product of the empty set of vectors in Eq. (\ref{eq:ls}) is to be understood as ${\bf 1},$ as in the example with $p=2.$ If $p$ is odd, and if $2k+1=p$, instead of tensor products of several vectors we have simply vectors, as in the example above with $p=3.$
\end{note*}
\begin{proof}
From Proposition \ref{prop:7} and from Lemma \ref{lem:par} we know that if $u\in \T^p(V)$ then only the components
$(\lambda_F(u))_{p-2k}$ for which $p-2k\geq 0$ can be non-vanishing. It remains to prove that for $p-2k\geq 0$ we have
\be
\begin{split}
&\left(\lambda_F(x_1\otimes\cdots\otimes x_p)\right)_{p-2k}=\\&\sum_{\pi\in\mathcal{P}_{p,k}}\sgn(\pi) F(x_{\pi(1)},x_{\pi(2)})F(x_{\pi(3)},x_{\pi(4)})\cdots\\
&\cdots F(x_{\pi(2k-1)},x_{\pi(2k)})\,x_{\pi(2k+1)}\otimes\cdots\otimes x_{\pi(p)}.
\end{split}
\label{eq:lfk}
 \ee
 The proof is by induction. We know that $\lambda_F(1)=\mathbf{1}$ and $\lambda_F(x)=x.$ Therefore
 \be \lambda_F(x_1\otimes x_2)=x_1\otimes x_2+i_{x_1}^F(x_2)=x_1\otimes x_2+F(x_1,x_2).
 \nn\ee
 Thus, in this case, we have $p=2$ and only one term of the type as in Eq. (\ref{eq:lfk}), namely for $k=1.$ Let us assume now that Eq. (\ref{eq:lfk}) holds for $p$ smaller than or equal to a given $p-1\geq 2$. We will show that then it is also valid for $p,$ i.e. that
 \be\begin{split}
 &\left(\lambda_F(x_0\otimes\cdots\otimes x_p)\right)_{p+1-2k}=\\
 &\sum_{\pi\in\mathcal{P}_{p+1,k}}\sgn(\pi) F(x_{\pi(0)},x_{\pi(1)})F(x_{\pi(2)},x_{\pi(3)})\cdots\\
&\cdots F(x_{\pi(2k-2)}),x_{\pi(2k-1)})\,x_{\pi(2k)}\otimes\cdots\otimes x_{\pi(p)},
\end{split}\label{eq:ls1}\ee
where $\mathcal{P}_{p+1,k}$ denotes the set of permutations of $\{0,1,2,\ldots,p\}$ satisfying the conditions suggested by Note \ref{note:sigma}.
We have
\be
\begin{split}
&(\lambda_F(x_0\otimes x_1\otimes\cdots\otimes x_p))_{p+1-2k}\nn\\&=x_0\otimes(\lambda_F(x_1\otimes\cdots\otimes x_p))_{p-2k}+i_{x_0}^F\left(\lambda_F(x_1\otimes\cdots\otimes x_p)_{p+2-2k}\right)\nn\\
\end{split}\ee
The first term
 in the formula above gives those terms in Eq. (\ref{eq:ls1}) for $p+1$ in which $x_0$ does not participate as an argument of $F$. We can extend the set of indices setting $\pi(0)=0,$ with the correct signs of the permutations (as $0$ must be transposed with an even number of indices). The second term gives the terms in Eq. (\ref{eq:ls1}) in which $x_0$ participates as an argument of $F$. In that case it must participate as the first argument, and the operator $i_{x_0}^F$ makes sure that this happens. In order to obtain
\be (\lambda_F(x_0\otimes\cdots\otimes x_p))_{p+1-2k}\nn\ee
we need to take
\be i_{x_0}^F((\lambda_F(x_1\otimes\cdots\otimes x_p))_{p+2-2k}).\nn\ee
From Eq. (\ref{eq:iii}) we know that $i_{x_0}^F$ commutes with $\lambda_F$, therefore we need to calculate
\be (\lambda_F(i_{x_0}^F(x_1\otimes\cdots\otimes x_p)))_{p+1-2k}.\nn\ee
We know the action of $i_{x_0}^F$
\be
\begin{split} &i_{x_0}^F(x_1 \otimes\cdots\otimes x_p)\nn\\&=\sum_{l=1}^p(-1)^{l-1}\,F(x_0,x_l)\,x_1\otimes\cdots\otimes \hat{x}_l\otimes\cdots\otimes x_p,\nn
\end{split}\ee
where the hat denotes the factor omitted in the product. Therefore we need to calculate
\be
\sum_{l=1}^p(-1)^{l-1}\,F(x_0,x_l)(\lambda_F(x_1\otimes\cdots\otimes \hat{x}_l\otimes\cdots\otimes x_p))_{p+1-2k}.\nn
\ee
For this term not to be automatically zero, we need to take $k\leq [(p-1)/2].$ We set $k'=k-1$, then $p+1-2k=(p-1)-2k'$.
According to the induction assumption we have
\be \begin{split}
&(\lambda_F(x_1\otimes\cdots\otimes \hat{x}_l\otimes\cdots\otimes x_p))_{p-1-2k'}\nn\\
&=\sum_{\pi\in\mathcal{P }_{p-1,k'}}\sgn(\pi)F(x_{\pi(1)},x_{\pi(2)})\cdots \\ &F(x_{\pi(2k'-1)},x_{\pi(2k')})\,x_{\pi(2k'+1)}\otimes\ldots\otimes x_{\pi(p)}\nn
\end{split}\ee
where $\mathcal{P }_{p-1,k'}$ now means a set of permutations of $\{1,2,\ldots,\hat{l},\ldots,p\}$, that is $\{1,2,\ldots,p\}$ without $l.$ Now, since $(-1)^{l-1}\sgn{\pi}$ is the signature of the permutation $(0,1,\ldots,p)\mapsto (0,l,\pi(1),\ldots,\widehat{\pi(l)},\ldots,\pi(p))$, the formula (\ref{eq:ls}) holds also for $x_0,\ldots,x_p.$
\end{proof}
\subsubsection{$\lambda_F$ as an exponential}\label{sec:lexp}
\begin{definition}\label{def:ak}
Using the same notation as in Proposition \ref{prop:ln} we define
\be a_0^F= \Id_{\T(V)},\nn\ee
and for $k>0$ and $p\geq 2k$
\be
 \begin{split}&a_k^F(x_1\otimes\cdots\otimes x_p)=\nn\sum_{{2\leq 2k\leq p}}\,\sum_{\pi\in\mathcal{P}_{p,k}}\sgn(\pi) F(x_{\pi(1)},x_{\pi(2)})F(x_{\pi(3)},x_{\pi(4)})\cdots\\
&\cdots F(x_{\pi(2k-1)},x_{\pi(2k)})\,x_{\pi(2k+1)}\otimes\cdots\otimes x_{\pi(p)}. \label{eq:ak}\end{split}\ee
For $p<2k$ we set $a_k^F(x_1\otimes\cdots\otimes x_p)=0.$
\end{definition}
It is clear from the Proposition \ref{prop:ln} and from the definition of $a_k$ that we have:
\be \lambda_F=\sum_{k=0}^\infty a_k^F.\label{eq:lsum}\ee
We will show now that the sequence $(a_k^F)_{k\in {\bf N}}$ satisfies the relations that are characteristic for an {\it exponential sequence\,} - as discussed, for example, in Ref. \cite[A IV 87]{bourbaki47}:
\begin{proposition}
With $a_k$ defined as in the Definition \ref{def:ak} we have
\be a_l^F a_k^F =\left(\begin{smallmatrix}k+l\\k\end{smallmatrix}\right)a_{k+l}^F=a_k^F a_l^F,\label{eq:dp}\ee
where $\left(\begin{smallmatrix}k+l\\k\end{smallmatrix}\right)$
is the binomial coefficient
\be \left(\begin{smallmatrix}k+l\\k\end{smallmatrix}\right) =\left(\begin{smallmatrix}k+l\\l\end{smallmatrix}\right)=\frac{(k+l)!}{k!l!}.\nn\ee
\end{proposition}
\begin{proof}
With $u\in \T^p(V)$, for $a_k^F(u)$ to be non-zero it is necessary that $p\geq 2k$, then $u'=a_k^F(u)$ is of grade $p'=p-2k$. For $a_l^F(u)'$ to be nonzero, we must have $p'-2l\geq 0$, i.e. $p\geq 2(k+l)$, as required by the statement. Assuming now $p\geq 2(k+l)$, we have
\be\begin{split}
 &a_l^F(a_k^F(x_1\otimes\cdots\otimes x_p))=\sum_{\pi\in\mathcal{P}_{p,k}}\sgn(\pi) F(x_{\pi(1)},x_{\pi(2)})F(x_{\pi(3)},x_{\pi(4)})\cdots\\
&\cdots F(x_{\pi(2k-1)},x_{\pi(2k)})\,a_l^F\left(x_{\pi(2k+1)}\otimes\cdots\otimes x_{\pi(p)}\right).\end{split}\label{eq:akl}\ee
Each of the terms in the expansion of $a_{k+l}^F(x_1\otimes\cdots\otimes x_p)$ is of the form
\be \sgn(\pi)F(x_{i_1},x_{j_1})\cdots F(x_{i_{k+l}},j_{k+l})X_{(i_1\ldots j_{k+l})},\nn\ee
where $i_1<i_2<\cdots <i_{k+l}$, $i_1<j_1,\ldots,i_{k+l}<j_{k+l},$ and $X_{(i_1\ldots j_{k+l})}$ stands for $x_1\otimes \cdots\otimes x_p$ with $x_{i_1},\ldots,x_{j_{k+l}}$ removed. Every such term can be obtained in $(k+l)!/(k!l!)$ ways from the terms
obtained in Eq. (\ref{eq:akl}) by selecting a subsequence of $k$ pairs $(i_m,j_m)$, and extending it to the whole  sequence of $(k+l)$ pairs via the action of $a^F_l$ as in Eq. (\ref{eq:akl}). All these terms will come with the required signature.
\end{proof}
\begin{rem}
It is clear from the proof above that, for any two bilinear forms $F,G$  the operators $a_k^F$ and $a_l^G$ commute. Therefore the family
$\{a_k^F\}$, with $k$ running through all natural integers and $f$ running through all bilinear forms generates a commutative subalgebra of $\End(\T(V)).$
\end{rem}
\begin{cor}
Assuming the characteristic of $\BK$ is $0$, and setting $a_F=a_1^F$, we have
\be \lambda_F = \exp(a_F).\ee
\end{cor}
\begin{proof}
From the definition and from Eq. (\ref{eq:dp}) we have
$$ (a_F)^2=a_1^F a_1^F=2!\, a_2^F,$$
$$(a_F)^3=a_F (a_F)^2=2!\, a_1^F a_2^F=2!\,\frac{3!}{2!}a_3^F=3!\,a_3^F,$$
and, in general,
\be (a_F)^n=n!\,a_n^F.\ee
Thus, from Eq. (\ref{eq:lsum}) we have
\be \lambda_F=\sum_{n=0}^\infty \frac{1}{n!}\,(a_F)^n=\exp(a_F).\nn\ee
The series is finite when applied to any element $u\in \T(V)$, because $\T(V)=\bigoplus_{p=0}^\infty \T^p(V)$, by the definition of the algebraic direct sum, consists of elements $u=\oplus_{p=0}^\infty\,  u_p$, $u_p\in \T^p(V)$, for which $u_p\neq 0$ only for a finite number of indices $p.$
\end{proof}
\subsubsection{The case of an alternating $F$ - the Pfaffian}
A bilinear form $F$ on $V$ is called alternating  if $F(x,x)=0$ for all $x\in V.$ If $F$ is alternating, then $F(x+y,x+y)=0$ for all $x,y\in V$, and therefore, by bilinearity, the form $F$ is antisymmetric: $F(x,y)=-F(y,x)$. Let $x_1,\ldots,x_p\in F.$ Then the matrix $f_{ij}=F(x_i,x_j)$ is alternating, i.e. $f_{ii}=0,\, f_{ij}=-f_{ji}.$

For an even-dimensional alternating matrix $f$ one defines the Pfaffian $\mbox{Pf}\,(f)$ as follows\footnote{See e.g. Ref. \cite[pp. 82-83]{bourbaki1959}.}
\begin{definition}
The Pfaffian of a $2n\times 2n$ alternating matrix $f=(f_{ij})$ is defined as
\be \mbox{Pf}\,(f)=\sum_\pi\, \sgn(\pi\,) f_{i_1j_1}f_{i_2j_2}\ldots f_{i_nj_n},\nn\ee
where the sum is over the set of all permutations $\pi$ on the set $\{1,2,\ldots,2n\}$ of the form
\be \pi=\begin{pmatrix}1&2&3&4&\cdots &2n\\i_1&j_1&i_2&j_2&\cdots &j_n\end{pmatrix},\nn\ee
for $i_k<j_k$ and $i_1<i_2<\cdots <i_n.$ The signature $\sgn(\pi)$ of the permutation $\pi$ is given by $(-1)^m$, where $m$ is the number of transpositions in $\pi.$
\label{def:pfa}\end{definition}
It is well known\footnote{See e.g. Ref. \cite[p.83]{bourbaki1959}.}  that the square of the Pfaffian of an even-dimensional alternating matrix is equal to its determinant. Comparing the Definition \ref{def:pfa} with the formula (\ref{eq:ls}) in Proposition \ref{prop:ln} and Definition \ref{def:ak}, we arrive at the following Corollary:
\begin{cor}
Let $F$ be an alternating bilinear form on $V$, let $x_1,\ldots,x_p\in V,$ with $p\geq 2,$ $p$ even, $p=2n.$  Then
\be \begin{vmatrix}F(x_1,x_1)&F(x_1,x_2)&\cdots &F(x_1,x_p)\\
F(x_2,x_1)&F(x_2,x_2)&\cdots &F(x_2,x_p)\\
.&.&\cdots &.\\
.&.&\cdots &.\\
.&.&\cdots &.\\
F(x_p,x_1)&F(x_p,x_2)&\cdots  &F(x_p,x_p)\end{vmatrix}=\omega^2,\nn\ee
where
\be \omega=(\lambda_F (x_1\otimes \cdots  \otimes x_p))_0=a_n^F(x_1\otimes \cdots  \otimes x_p).\nn\ee
\QED\end{cor}
 \section{Clifford algebras}
 \subsection{Quadratic forms}
 Given a module $M$ over a commutative ring $R$ there are two definitions possible of a quadratic form on $M$, one more general than the other if rings with an arbitrary characteristic are being considered.
 The following, more general, definition is used, in particular, in Bourbaki \cite{Bourbaki1998}, Chevalley \cite{chevalley1996}, and Helmstetter \& Micali \cite{Helmstetter2008}:
 \begin{definition}[Quadratic form I]\label{def:chev}
 Let $M$ be a module over a commutative ring $R.$ A mapping $Q:M\rightarrow R$ is called a quadratic form on $M$ if the following conditions are satisfied:
 \begin{enumerate}
 \item \be Q(\alpha x)=\alpha^2\,Q(x)\, \mbox{for all } \alpha\in R,\,x\in M,\label{eq:qa}\ee
 \item There exists a bilinear form $\Phi(x,y)$ on $M$ such that for all $x,y\in M$ we have
\be \Phi(x,y)=Q(x+y)-Q(x)-Q(y).\label{eq:beta}\ee
 \end{enumerate}
 We say that the bilinear form $\Phi$ is the  {\it polar form\,} associated with the quadratic form $Q.$ Sometimes $\Phi$ is also called {\em the polar form of $Q$}. It follows
 from its very definition that $\Phi$ is {\em symmetric}: $\Phi(x,y)=\Phi(y,x)$ for all $x,y\in M.$
 \end{definition}
We can combine Eqs. (\ref{eq:beta}) and (\ref{eq:qa}) into:
 \be Q(\alpha x+\beta y)=\alpha^2Q(x)+\beta^2 Q(y)+\alpha\beta\Phi(x,y).\label{eq:qab}\ee

 The short discussion of consequences given below is taken directly from Ref. \cite{Helmstetter2008}.

 \begin{note}From the very definition we find that \be\Phi(x,x)=Q(2x)-2Q(x)=4Q(x)-2Q(x)=2Q(x)\label{eq:q2}.\ee It follows that {\bf if $R$ is of characteristic 2}, then $\Phi(x,x)=0$ for all $x\in R.$  Such a form is called {\em alternating }. In that case, since also $\Phi(x+y,x+y)=0$,  we have that
 \be
 \begin{split}0=\Phi(x+y,x+y)=\Phi(x,x)+\Phi(x,y)+\Phi(y,x)+\Phi(y,y)\nn\\
 =\Phi(x,y)+\Phi(y,x),
 \nn\end{split}\ee
 so that in this case the form $\Phi$ is antisymmetric  $\Phi(x,y)=-\Phi(y,x).$
\end{note}

 Getting back to a general characteristic, we may also notice at this point that {\em if the mapping $x\mapsto 2x$ is surjective} in $M,$ then the form $\Phi$ determines $Q.$ Indeed, setting $y=2x$ we get $Q(y)=Q(2x)=4Q(x)=2\Phi(x,x).$ We also observe that the quadratic form $Q$ is  determined by the associated bilinear form $\Phi$ when the mapping $\alpha\mapsto 2\alpha$ is injective in $R,$ in particular if multiplication by $\frac{1}{2}$ makes sense in $R.$ In that case we can solve Eq. (\ref{eq:q2}) to obtain $Q(x)=\frac{1}{2}\Phi(x,x).$

 In applications to Clifford algebras, unless we are interested in very special cases like characteristic $2,$ it is more convenient to use a little bit different definition of a quadratic form, as given, for instance, in Ref. \cite[p. 199]{sulanke2008}:
 \begin{definition}[Quadratic form II]\label{def:sul}
 Let $M$ be a module over a commutative ring $R.$ A function $Q:M\rightarrow R$ is called a quadratic form if there exists a bilinear form $F:M\times M\rightarrow R$ such that
\be Q(x)=F(x,x).\label{eq:qsul}\ee
 \end{definition}
 It follows from this last definition that the condition in Eq.(\ref{eq:qa}) is then automatically satisfied, and also the condition in Eq.(\ref{eq:beta}) is automatically satisfied with \be \Phi(x,y)=F(x,y)+F(y,x).\label{eq:bb}\ee
\begin{rem}
If the module $M$ admits a basis (in particular, when it is a vector space), then given a quadratic form $Q$ as in Def. \ref{def:chev} one can always  construct a bilinear form $F$  (in general a non symmetric one) such that $Q(x)=F(x,x)$ (cf. Sec. \ref{sec:bf2} below).
\label{rem:qb}\end{rem}
\begin{definition}
If $F$ is a bilinear form on $V$, then the mapping $x\mapsto F(x,x)$ is a quadratic form. We will denote this form by $Q_F$:
\be Q_F(x)=F(x,x),\, x\in V.\nn\ee
\label{def:qb}\end{definition}
\subsubsection{Constructing a bilinear form in characteristic 2}\label{sec:bf2}
The construction here is taken from Ref. \cite[Proposition 2, p. 55]{Bourbaki2006}). \footnote{The construction can be also found in Ref. \cite[I.2.2, p.76]{chevalley1996}, but under the assumption that $V$ is finite dimensional.}

 Let $Q$ be a quadratic form on a vector space $V$ over a field $\BK.$ We start with noticing that $V,$ being a vector space, has a basis $\{e_i\}_{i\in I}.$ We select one such basis. By the well-ordering theorem every set can be well ordered, and we will assume that the index set $I$ is well ordered. Since $\{e_i\}_{i\in I}$ is a basis,  every bilinear form $F$ is uniquely determined by the coefficients $f_{ij}$, $i,j\in I.$ Let $\Phi$ be the bilinear form associated to $Q.$ We first observe that if $\{\alpha_i\}_{i\in I}$ is any family of elements of $\BK$ with only a finite number of $\alpha_i\neq 0,$ then
\be Q(\sum_i \alpha_i e_i)=\sum_i \alpha_i^2Q(e_i)+\sum_{\{i<j\}}\alpha_i\alpha_j\Phi(e_i,e_j),\label{eq:rec1}\ee
where the last sum is over all {\em two-element subsets $\{i,j\}$ of} $I.$ \footnote{Thus if $a_{i_1}$ and $a_{i_2}$ are nonzero, with $i_1<i_2$, then only $\Phi(e_{i_1},e_{i_2})$ enters the sum, and not $\Phi(e_{i_2},e_{i_1}),$ because $\{i_2,i_1\}$ is the same subset as $\{i_1,i_2\}.$ }

It is understood that each sum is over a finite set determined by non-zero $\alpha_i$-s. We prove Eq. (\ref{eq:rec1}) by induction with respect to the number $n$ of nonzero coefficients $\alpha_i$. If there are only two nonzero coefficients, then (\ref{eq:rec1}) follows from Eq. (\ref{eq:qab}), i.e. from the definition of the quadratic form \ref{def:chev}. Assume now that Eq. (\ref{eq:rec1}) holds for subsets $\{i_1,\ldots,i_n\}$ of $n$ non-zero coefficients $\alpha_i$, and let us add another non-zero coefficient $\alpha_{i_{n+1}}.$ Then
\begin{multline}
Q(\alpha_{i_1}e_{i_1}+\cdots +\alpha_{i_{n+1}}e_{i_{n+1}})=Q\left((\alpha_{i_1}e_{i_1}+\cdots +\alpha_{i_{n}}e_{i_{n}}\right)+\alpha_{i_{n+1}}e_{i_{n+1}})=\notag\\
Q(\alpha_{i_1}e_{i_1}+\cdots +\alpha_{i_n}e_{i_n})+Q(\alpha_{i_{n+1}}e_{i_{n+1}})+\Phi(\alpha_{i_1}e_{i_1}+\cdots +\alpha_{i_n}e_{i_n},\alpha_{i_{n+1}}e_{i_{n+1}})\notag.
\nn\end{multline}
Using now the quadratic form property, in particular for the sum of two elements, the assumed property for the sum of $n$ elements, as well as linearity of $\Phi$ in the first argument leads to the desired result. Nowhere do we need to assume that the basis $\{e_i\}_{i\in I}$ is finite.

Given a quadratic form $Q$ we can now define a bilinear form $F$ satisfying $Q(x)=F(x,x)$ by defining its coefficients $f_{ij},$ $i,j\in I,$ as follows:
\begin{eqnarray} f_{ii}&=&Q(e_i),\nn\\
 f_{ij}&=&\Phi(e_i,e_j),\, i<j,\nn\\
 f_{ij}&=&0,\, i>j.
 \end{eqnarray}
 We now check that $Q(x)=F(x,x)$ for every $x$ in $V.$ If $x\in V$ then $x=\sum_i\alpha_i e_i,$ with only a finite number of non-zero terms in the sum. Therefore, using Eq. (\ref{eq:rec1}) we have
\be Q(x)=\sum_i \alpha_i^2Q(e_i)+\sum_{i<j}\alpha_i\alpha_j\Phi(e_i,e_j).\label{eq:rec2}\ee
On the other hand, from bilinearity of $F$ we have
$$ F(x,x)=F(\sum_i \alpha_i e_i,\sum_j \alpha_j e_j)=\sum_i \alpha_i^2 f_{ii}+\sum_{i\neq j}\alpha_i\alpha_j f_{ij}=Q(x)$$
from the definition of the coefficients $f_{ij}$ above (because $ f_{ij}=0$ for $i>j$) .
\subsubsection{Action of bilinear forms on quadratic forms}\label{sec:abf}
Let $V$ be a vector space over $\BK$. The set of quadratic forms on $V$ is a vector space. We will denote it $\mbox{Quad }(V).$ We will denote $\bv(V)$, $\sv(V)$, $\av(V)$ the vector spaces of bilinear, symmetric, and alternating  forms respectively. If $Q\in \qv(V)$ and $F\in \bv(V)$ we will write $Q+F$ for the quadratic form $Q'(x)=Q(x)+Q_F(x).$ By Remark \ref{rem:qb} this action is transitive. If the characteristic of $\BK$ is $\neq 2$, then, given $Q,Q'$ there is a unique {\it symmetric\,} form $F$ such that $Q'=Q+F.$ On the other hand, in characteristic 2 there are quadratic forms $Q$ that are not of the form $Q=Q_F$, for $F\in \sv(V).$ \footnote{For a simple example of such a $Q$,  for $\BK=\BZ/\BZ_2$, in the two dimensional space $\BK^2$, see ref. \cite[p. 295, Example]{lounesto2001}.}
\subsection{Clifford algebra - definition}
 Let $V$ be a vector space over the (commutative) field $\BK$, let $Q$ be a quadratic form on $V$  (see Def. \ref{def:chev}), and let $I(Q)$ be the two-sided ideal in $\T(V)$ generated by elements of the form $x\otimes x-Q(x)\mathbf{1}$, where $x\in V\subset \T(V).$ The ideal $I(Q)$ consists of all finite sums of elements of the form $x_1\otimes\ldots\otimes x_p\otimes (x\otimes x-Q(x)\mathbf{1})\otimes y_1\otimes \ldots\otimes y_q,$ where $x,x_1,\ldots,x_p,y_1,\ldots ,y_q$ are in $V.$
\begin{definition}[Clifford algebra]\footnote{Here we follows the standard definition as in  \cite[p. 35]{chevalley1996}. An alternative definition can be found e.g. in \cite[p. 8]{gilbert91}}
With $V$ and $Q$ as above the quotient algebra $\cl(V,Q)=\T(M)/I(Q)$ is called the Clifford algebra associated to $V$ and $Q.$
\label{def:cli}\end{definition}
Denoting by $\pi_Q:\T(V)\rightarrow \cl(V,Q)$ the canonical mapping, $\pi_Q(V)$ is a subspace of $\cl(V,Q)$ that generates (together with $\pi_Q(\mathbf{1})$) $\cl(V,Q)$ as an algebra.

The case of $Q=0$ is special. The ideal $I(Q)$ is then generated by homogeneous elements $x\otimes x$ and the algebra $\cl(0)$ (i.e. $\cl(Q)$ for $Q=0$) is nothing but {\em the exterior algebra} $\bigwedge(V)$ of $V.$ All homogeneous elements of $I(0)$ are then of degree at least $2,$ therefore no non-zero element of $V$ can belong to this ideal. It then easily follows that in this case the mapping $x\mapsto \pi_Q(x)$ is an embedding and $V$ can be always identified with the grade 1 subspace of $\Lambda(V)=\cl(0)$.\label{sec:cad}
\begin{rem}[A note on the non-triviality of $I(Q).$]
How to make sure that $\mathbf{1}_{\T(V)}$ is not in $I(Q)$? This is easy for the exterior algebra $\bigwedge(V)=\cl(V,0)$, which is $\BZ$-graded and the ideal $I(Q)$ is generated by homogeneous expressions $x\otimes x.$ To show that $I(Q)\neq \T(V)$ in a general case, one can proceed as follows. Let $F$ be a bilinear form such that\footnote{Cf. Sec. \ref{sec:bf2}.} $Q(x)=F(x,x)$, and consider the linear map $L_F:V\rightarrow \mbox{End}(\bigwedge(B))$ defined as $L_F:x\mapsto \bar{\Lambda}_F(x)=\bar{e}_x+\bi_x^F,$ where $\bar{e}_x$ and $\bi_x^F$ are defined as in Sec.\ref{sec:ops}, except that we use them on $\bigwedge(V)$ instead of on $\T(V)$.\footnote{We will discuss these operations in extenso in the next section.} By the universality of the tensor algebra $\T(V)$,  $L_F$ extends uniquely  to an algebra homomorphism of unital algebras (denoted by the same symbol) $L_F:\T(V)\rightarrow \End(\bigwedge(V)).$ It is easy to see that $\bbe_x^2=(\bi_x^F)^2=0,$ and $\bbe_x\,\bi_x^F+\bi_x^F\,\bbe_x=F(x,x)=Q(x),$ and therefore $L_F$ vanishes on the ideal $I(Q)$. But, on the other hand, $L_F(\mathbf{1})=\Id_{\bigwedge(V)},$ therefore $\mathbf{1}_{\T(V)}\notin I(Q),$ and $\pi_Q(\mathbf{1}_{\T(V)})$ can be identified with the unit element of $\BK\subset \cl(Q).$ In the following we will often simply write $\mathbf{1}$ instead of $\pi_Q(\mathbf{1})$ or $\mathbf{1}_{\cl(Q)}.$

Moreover, since $L_F(x)(\mathbf{1}_{\bigwedge(V)})=x\in\bigwedge(V),$ it follows that the mapping $V\ni x\mapsto \pi_Q(x)$ is injective.
\label{rem:1}\end{rem}

The following proposition is an immediate consequence of the universal property of the tensor algebra and of the definition of the Clifford algebra above (c.f. \cite[Proposition 1, p. 140]{bourbaki1959}).
\begin{proposition}[Universal property]\label{pro:up}
Let $\phi$ be a linear map from $V$ into an algebra $\mathcal{A}$ (with unit $\mathbf{1}$), such that
$\phi(x)^2=Q(x)\mathbf{1}$ for all $x\in V.$ Then $\phi$ can be extended to a unique algebra homomorphism $\bar{\phi}:\cl(Q) \rightarrow \mathcal{A}$. If $\mathcal{A}$ is $\BZ_2$-graded, and $\phi$ maps $V$ into an odd subspace of $\mathcal{A}$, then $\bar{\phi}$ is a homomorphism of $\BZ_2$-graded algebras.
\end{proposition}
\begin{proof}
As mentioned above, the first part of the proposition follows from the universal property of the tensor algebra and the definition \ref{def:cli} of $\cl(Q).$ The second part follows from the fact that $V$ generates $\cl(V),$ and that the even (resp. odd) part of $\cl(V)$ is generated by products of even (resp. odd) number of elements of $V.$
\end{proof}
The ideal $I(Q)$ is stable under the main involution $\alpha$ and the main anti-involution $\tau$ (Eqs. (\ref{eq:alpha}) and (\ref{eq:tau})). Therefore $\alpha$ and $\tau$ descend to the quotient Clifford algebra $\cl(Q)$. We shall use the same symbols $\alpha,\tau$ to denote the induced involution and anti-involution on $\cl(Q).$ While the tensor algebra $\T(V)$ is $\BZ-$graded, the Clifford algebras $\cl(Q)$ for $Q\neq 0$ are only $\BZ_2$-graded:
\be \cl(Q)=\cl^+(Q)\oplus \cl^-(Q), \nn\ee
where $\cl^\pm(Q)$ are the images of even/odd parts of the tensor algebra. If the mapping $x\mapsto 2x$ of $\cl(Q)$ to itself is injective, then
\be \cl^{\pm}(Q)=\{u\in \cl(Q):\alpha(u)=\pm u\}.\nn\ee
Alternatively $\cl^+(Q)$ (resp. $\cl^-(Q))$ can be defined as the linear subspace of $\cl(Q)$ generated by the products
$x_1\cdots x_p$, $x_i\in V$, with $p$ even (resp. odd). We have
\be
\begin{split}\cl(Q)^+\cl(Q)^+\subset \cl(Q)^+,\,\cl(Q)^+\cl(Q)^-\subset \cl(Q)^-,\,\\ \cl(Q)^-\cl(Q)^-\subset \cl(Q)^+.
\end{split}\label{eq:eo}\ee
In particular $\cl(Q)^+$ is a subalgebra of $\cl(Q).$
It is called {\em the even Clifford algebra}.
\section{Natural operations within and between Clifford algebras}
In this section, the main section of this paper, we will study the natural linear operations acting  within and between Clifford algebras over the same vector space $V.$ Connecting Clifford algebras $\cl(V,Q)$ and $\cl(V,Q')$ can be achieved by operations at the higher level, namely at the level of the tensor algebra $\T(V)$ - the one source of all $\cl(V,Q)$'s,  through the quotient maps.
\subsection{Anti-derivations $\bi_f$.}
\label{sec:ibf}
\begin{lemma}\label{lem:ibf}
If $Q$ is a quadratic form on $V$ then the ideal $I(Q)$ is stable under $i_f$, that is $i_f(I(Q))\subset I(Q),$ and thus $i_f$ defines an endomorphism $\bi_f$ on the Clifford algebra $\cl(Q)=\T(V)/I(Q)$:
\be \pi_Q\circ i_f=\bi_f\circ\pi_Q.\label{eq:ifb}\ee
In particular
\be \bi_f(\mathbf{1})=0.\nn\ee
\end{lemma}
\begin{proof}
Let $\mathcal{S}$ denote the set of all $u\in I(Q)$ for which $i_f(u)\in I(Q).$ Using Eq. (\ref{eq:ift}) we instantly get that $\mathcal{S}$ is a left ideal. On the other hand, if $u=(x\otimes x -Q(X))\otimes v$, then a simple calculation shows that
\be
i_f(u)=(x\otimes x-Q(x))\otimes i_f(u),\ee
and therefore $\mathcal{S}$ contains the right ideal generated by $(x\otimes x-Q(x)).$ It follows that $I(Q)\subset \mathcal{S}$, and thus $I(Q)$ is stable under $i_f,$ so that $i_f$ defines the linear mapping $\bi_f$ on the quotient algebra $\cl(Q)$, with the property
\be \bi_f\circ\pi_Q=\pi_Q\circ i_f.\label{eq:ibf}\ee
In particular, by Eq. (\ref{eq:if1}), we have $\bi_f(\mathbf{1})=0.$
\end{proof}
The operators $e_x:\T(V)\rightarrow \T(V),\,x\in V,$ of left multiplication in the tensor algebra evidently map the ideal $I(Q)$ into itself, and therefore descend to the quotient $\cl(Q)=\T(V)/I(Q)$, and become the operators of left multiplication in $\cl(Q)$
We will denote these quotient operators $\bar{e}_x:$
\be \pi_Q\circ e_x=\bar{e}_x\circ \pi_Q.\label{eq:expi}\ee
It follows from the very definition that for $x,y\in V$
\be \bar{e}_x^2=Q(x),\ \bar{e}_x\bar{e}_y+\bar{e}_y\bar{e}_x= \Phi(x,y).\label{eq:bex}\ee
\begin{proposition}
 On $\cl(Q)$ we have\footnote{The right hand sides in Eqs. (\ref{eq:bex}),(\ref{eq:iepei}) are to be understood as multiplication by scalars in $\mbox{End}(\cl(Q))$. }
\begin{enumerate}
\item[{\rm (i)}] $ \bi_f(\mathbf{1})=0,\,(\mathbf{1}\in \cl(Q))$
\item[{\rm (ii)}] For all $x \in V,\,f,g\in V^*$ we have
\be \bi_f^2=0,\,\bi_f\bi_g+\bi_g\bi_f=0,\label{eq:bif2}\ee
\be \bi_f\bar{e}_x+\bar{e}_x\bi_f=f(x),\label{eq:iepei}\ee
or, explicitly: for all $x\in V$, $w\in \cl(Q),$ we have
\be \bi_f(x w)=f(x)\,w-x\, \bi_f(w).
    \label{eq:ifp}\ee
\item[{\rm (iii)}] For $x_1,\ldots,x_p\in V$  we have
\be \bi_f(x_1\cdots x_p)=\sum_{i=1}^p (-1)^{i-1}f(x_i)\,x_1\cdots \hat{x}_i\cdots x_p.\label{eq:ifx2p}\ee
\end{enumerate}
\label{prop:if2}\end{proposition}
\begin{proof}
We have already established (i). The formulas in (ii) follows immediately by applying $\pi_Q$ to Eqs. (iii) and (iv) in Proposition \ref{prop:if}, and  to Eq. (\ref{eq:ifx1p}), and taking into account the fact that $\pi_Q$ is an algebra homomorphism.\footnote{The formula (\ref{eq:ifx2p}) is usually proven for the $\BZ$-graded exterior algebra $\bigwedge(V)$ rather than for a general Clifford algebra $\cl(Q)$ which is only $\BZ_2$ graded. See e.g. \cite[Exerc. 5, p. 155]{Bourbaki2006}.}
\end{proof}
As it was in the case with $\bi_f$, we will denote by $\bi_x^F$ the antiderivation $\bi_f$ for $f(y)=F(x,y):$
\be \bi_x^F=\bi_f, \quad \mbox{for } f(y)=F(x,y),\,(x,y\in V).\nn\ee
The mapping $V\ni x\rightarrow \bi_x^F\in \mbox{End}\left(\cl(Q\right))$ is linear and has the property
$(\bi_x^F)^2=0.$ Therefore it extends to a unique algebra homomorphism $\bigwedge(V)\ni u\mapsto \bi_u^F\in\mbox{End}(\cl(Q)).$ In particular, if $Q=0$, we have the map $u\mapsto \bi_{u}^F,$ $\bigwedge (V)\rightarrow \mbox{End}\left(\bigwedge(V)\right)$, sometimes written using the contraction symbol $\rfloor$:
\be u\JJF v\doteq\bi_u^F(v).\nn\ee
\subsection{The mapping $\bl_F$}
The Proposition below encapsulates one of the most important properties of $\lambda_F$. For an alternative proof cf. \cite[Proposition 3, p. 143]{Bourbaki2006}.
\begin{proposition}
Let $Q$ and $Q'$ be two quadratic forms on $V$ such that $Q'(x)=Q(x)+F(x,x),$  where $F(x,y)$ is a bilinear form. The mapping $\lambda_F$ maps the ideal $I(Q')$ onto the ideal $I(Q)$, $\lambda_F(I(Q'))=I(Q),$ and therefore it defines a linear isomorphism, denoted $\bl_F$ of  $\cl(Q')$ onto $\cl(Q)$:
\be \pi_Q\circ \lambda_F=\bl_F\circ\pi_{Q'}.\label{eq:blam}\ee
\label{prop:ll}\end{proposition}
\begin{proof}
To show that $\lambda_F(I(Q'))\subset I(Q)$ it is enough to show that $\lambda_F(u\otimes (x\otimes x-Q'(x))\otimes v)$ is in $I(Q)$ for all $u,v\in\T(V),\, x\in V.$ Using the fact that $\Lambda_F$ is an algebra homomorphism $\Lambda_F:\T(V)\rightarrow \End(\T(V))$ and Eq. (\ref{eq:L1}) we have
\be \lambda_F(u\otimes(x\otimes x-Q'(x))\otimes v)=\Lambda_F(u)\big(\Lambda_F(x\otimes x-Q'(x))(\lambda_F(v))\big).\ee
A straightforward calculation, using the properties $2$ and (iii) of Proposition \ref{prop:if}, gives
\be \Lambda_F(x\otimes x-Q'(x))=e_x\circ e_x-Q(x)\Id,\ee
where $\Id$ denotes the identity endomorphism of $\T(V).$ It immediately follows that $\Lambda_F(x\otimes x-Q'(x))(\lambda_F(u))=(x\otimes x-Q(x))\otimes\lambda_F(u)$ is in $I(Q).$ Now it suffices to prove that $\Lambda_F(u)$ maps $I(Q)$ into $I(Q).$ It suffices to do it when $u$ is a vector $y$ in $V,$ and $\Lambda_F(y)=e_y+i_y^F$. On one side, since $I(Q)$ is an ideal, it is evident that $e_y(I(Q))\subset I(Q)$; on the other side, we know from Lemma \ref{lem:ibf} that $i_y^F(I(Q))\subset I(Q)$. All this proves that $I(Q)\subset \lambda_F(I(Q')),$ and since $\lambda_{-F}=\lambda_F^{-1},$ it follows that $I(Q)\subset\lambda_F(I(Q')),$ and finally $\lambda_F(I(Q'))=I(Q).$

\end{proof}
\begin{proposition}
For all $x\in V$, $w\in \cl(Q')$ we have
\begin{eqnarray}
\bl_F(\mathbf{1})&=&\mathbf{1},\label{eq:lb1}\\
\bl_F(x)&=&x,\label{eq:lb2}\\
\bl_F(xw)&=&\bi_x^F(\bl_F(w))+x\bl_F(w),\label{eq:lb3}
\end{eqnarray}
where in Eq. (\ref{eq:lb3}) the multiplication $xw$ on the left is in the algebra $\cl(Q'),$ while the multiplication $x\bl_F(w)$ on the right is in the algebra $\cl(Q).$ Equations (\ref{eq:lb1}) and (\ref{eq:lb3}) define $\bl_F$ uniquely on $\cl(Q).$

For $p\geq2$ and $x_1,\ldots,x_p\in V$, we have\footnote{For the notation convention cf. {\bf Note} before the proof of Proposition \ref{prop:ln}.}
\be
\begin{split}
&\lambda_F(x_1\cdots x_p)=x_1\cdots x_p\\
&+\sum_{{2\leq 2k\leq p}}\,\sum_{\pi\in\mathcal{P}_{p,k}}\sgn(\pi) F(x_{\pi(1)},x_{\pi(2)})F(x_{\pi(3)},x_{\pi(4)})\cdots\\
&\cdots F(x_{\pi(2k-1)},x_{\pi(2k)})\,x_{\pi(2k+1)}\cdots x_{\pi(p)}.
\end{split}
\label{eq:lsq}\ee
 where the $\mathcal{P}_{p,k}$ is as in Definition \ref{def:perm}.

If $F,G$ are bilinear forms, if $Q''(x)=Q'(x)+G(x,x)$ and $Q'(x)=Q(x)+F(x,x),$ then
\be \bl_{F+G}=\bl_F\circ\bl_G.\label{eq:lbc}\ee:
\be
\begin{tikzcd}[column sep=1.em]
 & \cl(Q') \arrow{dr}{\bl_F} \\
\cl(Q'') \arrow{ur}{\bl_G} \arrow{rr}{\bl_{F+G}} && \cl(Q)
\end{tikzcd}
\label{eq:lbcd}\ee
We have
\be {(\bl_F)}^{-1}=\lambda_{(-F)}\label{eq:lfmin},\ee
therefore, in particular, $\bl_F$ is injective.

\label{prop:bl}\end{proposition}
\begin{proof}
Eqs. (\ref{eq:lb1}) and (\ref{eq:lb3}) follows immediately by applying $\pi_Q$ to both sides of Eqs. (\ref{eq:l1}) and  (\ref{eq:lb}) and making use of Eqs. (\ref{eq:ifp}) and (\ref{eq:blam}), with $w=\pi_Q(u).$ Eq. (\ref{eq:lbc}) follows directly from Eq. (\ref{eq:lbc0}) and the definition of the quotient mappings, as we have the general property of compositions of quotients of mappings. Eq. (\ref{eq:lsq}) follows directly from Eq. (\ref{eq:ls}) by applying $\pi_Q$. \end{proof}

\begin{cor}
If $x_1,\ldots,x_n$ are in $V$ and if $F(x_i,x_j)=0$ for $i<j$, then
\be \bl_F(x_1\cdots x_n)=x_1\cdots x_n,\nn\ee
where the product in the argument of $\bl_F$ is taken in $\cl(Q')$ and the product on the right hand side is in $\cl(Q).$
\label{cor:sqo}\end{cor}
\begin{proof}
The statement follows directly from Eq. (\ref{eq:lsq}).
\end{proof}
\subsubsection{The mapping $\bL_F$}\label{sec:bLF}
The mapping $\bL_F$ defined below is a straightforward generalization of the Chevalley's construction, cf. \cite[p. 70]{chevalley1996}.

Let $Q$ be a quadratic form and $F$ a bilinear form on a vectors space $V$ over a field $\BK$. With $x\in V$ let $\bL_x^F\in \End\,(\cl(Q))$ be defined as
\be \bL_x^F = \bar{e}_x+\bi_x^F.\nn\ee
Using Eqs. (\ref{eq:bex}),(\ref{eq:bif2}),(\ref{eq:iepei}) we find that
\be (\bL_x^F)^2=Q'(x)\,\mathbf{1}_{\cl(Q)},\nn\ee
where $Q'(x)$ is the quadratic form
\be Q'(x)=Q(x)+F(x,x).\nn\ee
Therefore, by the universal property (see Proposition \ref{pro:up}),  $\bL_x^F$ extends to a unique algebra homomorphism, denoted $\bL_F,$
\be \bL_F:\cl(Q')\rightarrow \End(\cl(Q)).\nn\ee
Since every Clifford algebra has a natural $\BZ_2$-gradation , the algebraf $\End(\cl(Q))$ is also naturally $\BZ_2$-graded. In general, if $W_1,W_2$ are graded vector spaces, the gradation  of the space $\mbox{Hom }(W_1,W_2)$ is given by
\be \mbox{Hom}^k(W_1,W_2)=\{\phi\in \mbox{Hom }(W_1,W_2):\phi(W_1^i)\subseteq W_2^{i+k}\}.\nn\ee
Since $\bL_x^F$ is an odd element of $\End(\cl(Q)$, it follows that $\bL_F,$ is also a homomorphism of $\BZ_2$-graded algebras $\bL_F: \cl(Q')\rightarrow \End(\cl(Q)).$ It follows from the definition that
for $x_1,\ldots,x_p\in V$ we have
\be \bL_F(x_1\cdots x_p)=(\bar{e}_{x_1}+\bi^F_{x_1})\cdots (\bar{e}_{x_p}+\bi^F_{x_p}).\nn\ee
From Eqs. (\ref{eq:L1p}),(\ref{eq:ifb}),(\ref{eq:expi}) we have the relation between $\bL_F$ and $\Lambda_F$, namely, for all $F\in \mbox{Bil}(V), Q\in \mbox{Quad}(V),\, u,v\in \T(V)$ we have
\be \bL_F(\pi_{Q'}(u))(\pi_Q(v))=\pi_Q(\Lambda_F(u)(v)).\label{eq:LbL}\ee
From Eq. (\ref{eq:LFG}) we obtain now
\be \bL_F\left(\bL_G(u)(v)\right)(w)=\bL_{F+G}(u)\left(\bL_G(v)(w)\right),\label{eq:bLFG}\ee
for all $F,G\in \mbox{Bil}(V),Q\in\mbox{Quad}(V),u\in\cl(Q+Q_{F+G}),v\in\cl(Q+Q_F),$ $w\in\pi_Q(\T(\Rad(G)))=\cl(\Rad(G),Q).$

Let us recall that from Eqs. (\ref{eq:L1p}),(\ref{eq:lf1}) we have
\be \lambda_F(x_1\otimes\cdots\otimes x_p)=\Lambda_F(x_1\otimes\cdots\otimes x_p)(1)=(e_{x_1}+i^F_{x_1})\cdots (e_{x_p}+i^F_{x_p})(\mathbf{1}).\nn\ee
Applying $\pi_Q$ to both sides we obtain:
\be \bl_F(x_1\cdots x_p)=\bL_F(x_1\cdots x_p)(\mathbf{1}),\nn\ee
or \be \bl_F(u)=\bL_F(u)(\mathbf{1})\quad\mbox{for all}\quad u\in \cl(Q').\label{eq:lL1}\ee
We thus obtain the following commutative diagram
\[
\begin{tikzcd}[column sep=1.em]
 & \End(\cl(Q)) \arrow{dr}{\mbox{ev}(\cdot,1)} \\
\cl(Q') \arrow{ur}{\bL_F} \arrow{rr}{\bl_F} && \cl(Q)
\end{tikzcd}
\]
where $\mbox{ev}$ is the evaluation map. Since, (cf. Proposition \ref{prop:bl}) $\bl_F$ is injective, it follows that $\bL_F$ is injective as well, that is that $\bL_F$ is a graded algebra isomorphism from $\cl(Q)$ onto its image in $\End(\cl(Q')).$

\subsubsection{A comment on representations of Clifford algebras $\cl(V,Q)$ on $\bigwedge(V)$}\label{sec:reps}
Let $F$ be a bilinear form on $V,$ and let $Q=Q_F$, i.e. $Q(x)=F(x,x),\,x\in V.$ The mapping $\bL_F$, defined above in Section \ref{sec:bLF}, is then an algebra homomorphism from $\cl(V,Q)$ to $\End(\bigwedge(V))$. In other words $\bigwedge(V)$ becomes a Clifford module, and $u\mapsto \bL_F(u)$ becomes a (a faithful) representation of the algebra $\cl(V,Q)$ on $\bigwedge(V).$ Let us denote this representation $\rho_F$:
\be \rho_F(u)=\bL_F(u).\nn\ee

If $F'$ is another bilinear form with the property $Q_{F'}=Q_F=Q,$ then $F'=F+A,$ where $A$ is an alternating  form. In this case we have two representations, $\rho_F$ and $\rho_{F'}$  of the same algebra $\cl(V,Q)$ on $\bigwedge(V).$

Let us recall (see e.g. \cite[p. 49]{erdmann}) that two representations $\rho$ and $\rho'$ of the same algebra $\mathcal{A}$ on a vector space $\mathcal{W}$ are said to be equivalent if there is a vector space automorphism $\psi:\mathcal{W}\rightarrow \mathcal{W}$ such that
\be \rho'(a)=\psi\circ\rho(a)\circ\psi^{-1}\quad \mbox{for all}\quad a\in \mathcal{A}.\nn\ee
\begin{proposition}
With the assumptions as above, the representations $\rho_F$ and $\rho_{F'}$ of $\cl(V,Q)$ on $\bigwedge(V)$ are equivalent. Namely, we have
\be \rho_{F'}(a)=\lambda_A\circ\rho_{F}(a)\circ\lambda_A^{-1},\quad \mbox{for all}\quad  a\in \cl(V,Q).\nn\ee
\end{proposition}
\begin{proof}
Indeed, replacing $(F,G,w)$ with $(A,F,\mathbf{1})$ in Eq. (\ref{eq:bLFG}), and  using Eq. (\ref{eq:lL1}) we obtain
\be \bl_A\left(\bL_F(u)(v)\right)=\bL_{F'}(u)\left(\bl_A(v)\right),\nn\ee
or
\be \bl_A\circ\rho_F(u)=\rho_{F'}(u)\circ\bl_A.\nn\ee
\end{proof}
\subsection{Automorphisms and deformations in the bundle of Clifford  algebras}\label{sec:aut}
We have arrived at the following picture: We have an action, let us denote it by  $\tilde{\lambda}$, of the additive group of bilinear forms $\bv(V)$  on the manifold of quadratic forms $\qv(V)$,
the stability subgroup being $\av(V)$ - the additive group of alternating  forms:
\begin{eqnarray}
\tilde{\lambda}:\bv(V)\times\qv(V)&\rightarrow& \qv(V)\notag\\
\tilde{\lambda}(F,Q)&=&Q',\\ Q'(x)&=&Q(x)+F(x,x)\notag.
\end{eqnarray}
The group $\bv(V)$ acts on the basis of the $\BZ_2$-graded vector bundle $\cl(V)$ whose fibers are Clifford algebras $\cl(Q).$ And we know that this action admits what
is called a {\em lifting}, and we denote it with the letter $\tilde{\lambda}$, to the bundle $\cl(V):$
\be \tilde{\lambda}(F,u)=\bl_{-F}(u),\, u\in \cl(Q),\nn\ee
where $\bl_F$ have been defined in Proposition \ref{prop:ll}.\footnote{By definition the mapping $\bl_F$ is from $\cl(Q')$ to $\cl(Q)$. To map $\cl(Q)$ to $\cl(Q')$, where $Q'=Q+Q_F,$ we need to take the inverse map $\bl_{-F}=\bl_F^{-1}.$} Now $\tilde{\lambda}(F,\cl(Q))=\cl(Q').$ Thus fibers are mapped onto fibers by linear isomorphisms - see Fig. \ref{fig:bca}. For $F\in\av(V)$ we have $Q'=Q$, and so each fiber $\cl(Q)$ is mapped linearly onto itself. The linear morphisms above are also preserving the natural $\BZ_2$-gradations of the Clifford algebras..
\begin{figure}
\begin{tikzpicture}[scale=1.5]
\draw (1,2) -- (7,2) -- (7,5) -- (1,5) -- (1,2);
\draw (1,1) -- (7,1);
\draw (2,2) -- (2,5);
\draw (2,0.9) -- (2,1.1);
\draw (5,0.9) -- (5,1.1);
\node at (6.5,4.7) {\large{$\cl(V)$}};
\node at (8,1) {\large{$\qv(V)$}};
\node at (2,0.5) {\large{$Q$}};
\node at (5,0.5) {\large{$Q'$}};
\draw (5,2) -- (5,5);
\node at (5.5,3.5) {\large{$\cl(Q')$}};
\node at (2.5,3.5) {\large{$\cl(Q)$}};
\draw [->] (2.5,0.5) -- (4.5,0.5);
\node at (3.5,0.75) {\large{$\tilde{\lambda}$}};
\draw [->] (2,2.5) -- (5,3);
\node at (3.5,3) {\large{$\bl$}};
\node at (3.5,4.5) {\large{$\bl$}};
\draw [->] (2,4) -- (5,4.5);

\end{tikzpicture}
\caption{Every bilinear form $F\in \bv(V)$ defines an automorphism of the $\BZ_2$-graded vector bundle of Clifford algebras mapping linearly fibers onto fibers. Alternating  forms in $\av(V)\subset\bv(V)$ define vertical automorphisms - they do not move points on the base and map every fiber into itself. Such automorphisms are also called {\em gauge transformations} }
\label{fig:bca}\label{fig:1}\end{figure}
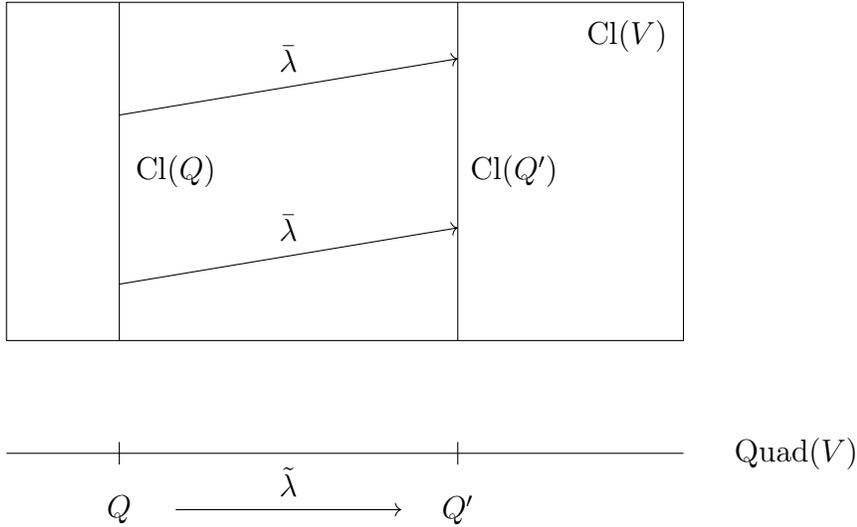

In characteristic $\neq 2$ the $\BZ_2$-graded vector bundle $\cl(V)$ admits a natural parallelization. If $Q$ and $Q'$ are any two quadratic forms, then there is a unique symmetric bilinear form F, such that $Q'=Q+Q_F,$ where $Q_F$ is given by Definition \ref{def:qb}, namely $F=\frac12 (\Phi'-\Phi),$ where $\Phi, \Phi'$ are given by Eq. (\ref{eq:beta}). In that case $\bl_F$ provides a natural $\BZ_2$-graded vector space isomorphism between $\cl(V,Q')$ and $\cl(V,Q).$ In characteristic $2$ the bundle $\cl(V)$ is also parallelizable, and the parallelization can be achieved using a non-symmetric bilinear form $F$ defined, for instance, by the construction in Section \ref{sec:bf2}. In that case the parallelization depends on the choice of a basis in $V.$

Given two quadratic forms $Q$ and $Q'$, such that $Q'=Q+Q_F$, we can use the mapping $\bl_F$ to realize the Clifford algebra $\cl(Q')$ as a deformation of the Clifford algebra product in $\cl(Q)$ as follows. For $u,v\in \cl(Q)$ let $uv$ stand for the product of $u$ and $v$ in $\cl(Q).$ Then $\bl_{-F}(u)$ and $\bl_{-F}(v)$ are in $\cl(Q')$. We multiply them in $\cl(Q')$ and transform their product back to $\cl(Q)$ using $\bl_F.$ This way we obtain the $\cl(Q'),$  product, with $Q'=Q+Q_F$, realized within $\cl(Q)$:
\be \BFprod{u}{v}=\bl_F\left(\bl_{-F}(u)\,\bl_{-F}(v)\right).\label{eq:uvf1}\ee
Since $V$ generates any Clifford algebra $\cl(V,Q)$, in particular $\cl(Q'),$ it is enough to have the formula (\ref{eq:uvf1}) for $u=x\in V.$ In that case we can use the fact that $\bl_{-F}(x)=x$ - see Eqs. (\ref{eq:lb2}),(\ref{eq:lb3}) - to obtain
\begin{eqnarray}
\BFprod{x}{v}&=&\bl_F\left(\bl_{-F}(x)\,\bl_{-F}(v)\right)\notag\\
&=& \bl_F\left(x\,\bl_{-F}(v)\right)=x\,\bl_F\left(\bl_{-F}(v)\right)+\bi_x^F\left(\bl_F(\bl_{-F}(v))\right)\notag\\
&=&xv+\bi_x^F(v)\label{eq:xFv}.
\end{eqnarray}

\subsubsection{The action of alternating  forms}
In this section we will assume that $V$ is a finite-dimensional vector space over an arbitrary field $\BK$. Later on in this section, when discussing the exponential, we will additionally assume that $\BK$ is of characteristic zero.

There is a natural linear mapping from $\bigwedge^2(V^*)$ to $\av(V)$ defined by
\be \langle f\wedge g,x\wedge y\rangle=g(x)f(y)-f(x)g(y),\label{eq:fwg}\ee
for $f,g\in V^*$ and $x,y\in V$. Since we have assumed that $V$ is finite-dimensional, the mapping above is a bijection - cf. e.g \cite[p. 591, Proposition 7]{Bourbaki1998}. Every element $A^*\in \bigwedge^2(V^*)$ defines an alternating  form $A\in\av(V)$:
\be A(x,y)=\langle A^*,x\wedge y\rangle.\label{eq:axy}\ee
In particular, for $A^*=f\wedge g$ we have
\be A(x,y)=g(x)f(y)-f(x)g(y).\label{eq:Axy}\ee
From Section \ref{sec:ibf} we know that we have a linear map $V^*\ni f\mapsto \bi_f\in \End(\cl(V,Q))$ that associates to every linear form $f\in V^*$ the antiderivation $\bi_f\in \End(\cl(V,Q))$ with the property $\bi_f^2=0$. From the universal property of the exterior algebra it then follows that this mapping has a unique extension to an algebra homomorphism
\be \bigwedge(V^*)\ni u^*\mapsto i_{u^*}\in \End(\cl(V,Q)),\nn\ee
thus making $\cl(V,Q)$ a left $\bigwedge(V^*)$ module. In particular, for $f,g\in V^*$ we have
\be \bi_{f\wedge g}=\bi_f\circ\bi_g.\label{eq:ibfg}\ee
One often writes $u^*\lrcorner v$ for $u^*\in \bigwedge(V^*)$, $v\in \cl(V,Q),$ and calls it the interior product of $u^*$ and $v$ - cf. e.g. \cite[Ch. 4.4]{Helmstetter2008}. We will be interested in the action of $i_{A^*}=A^*\lrcorner \,\cdot\,$ for $A^*\in \bigwedge^2(V^*).$ It is evident that $\bi_{A^*}(\mathbf{1}_{\cl(V,Q)})=0.$ On the other hand, with $\bi_x^F$ defined  as in Section \ref{sec:ibf} we have:
\begin{proposition}
For all $x\in V$, $u\in\cl(V,Q)$ we have
\be i_{A^*}(xu)=x\,\bi_{A^*}(u)+\bi_x^A(u).\label{eq:prA}\ee
\end{proposition}
\begin{proof}
Let us first consider the case of $A^*$ being a simple two-form $A^*=f\wedge g$, $ f,g\in V^*,$ in which case we can use Eq. (\ref{eq:fwg}) relating $A^*\in \bigwedge^2(V^*)$  to $A\in \av(V).$ In this case we have
\begin{eqnarray} \bi_{A^*}(xu)&=&\bi_f(\bi_g(xu))=\bi_f(g(x)u-x\bi_g(u))\nn\\
&=&g(x)\bi_f(u)-f(x)\bi_g(u)+x\bi_f(\bi_g(u))\notag\\
&=&g(x)\bi_f(u)-f(x)\bi_g(u)+x\,\bi_{A^*}(u).\nn\end{eqnarray}
We will now show that
\be g(x)\bi_f(u)-f(x)\bi_g(u)=\bi_x^A(u).\label{eq:gxfu}\ee
To show (\ref{eq:gxfu}), let us fix $x$ and write (\ref{eq:Axy}) as
$$ A(x,y)=g(x)f(y)-f(x)g(y)=h(y),\quad\mbox{where}\quad h=g(x)f-f(x)g.$$
Consequently
$\bi_x^A=\bi_h=g(x)\bi_f-f(x)\bi_g,$ as in (\ref{eq:gxfu}).

Thus, for simple $A^*=f\wedge g$ we have obtained
\be \bi_{A^*}(xu)=\bi_x^A(u)+x\bi_{A^*}(u).\label{eq:iaxu}\ee
But then Eq. (\ref{eq:iaxu}) extends by linearity to the whole $\bigwedge^2(V^*).$
\end{proof}

It will be convenient to write Eq. (\ref{eq:prA}) as a commutator
\be [\bi_{A^*},\bar{e}_x]=\bi_x^A,\label{eq:cie}\ee
where $\bar{e}_x$ is the operator of left multiplication by $x\in V$ - c.f Eq. (\ref{eq:expi}).

\noindent The following Corollary is a stripped down and application oriented version of a much more sophisticated and general result to be found in Ref. \cite[Theorem 4.7.13, p. 208]{Helmstetter2008}. Here we use the ordinary exponential of a linear operator, as it is done in physics oriented formulations. Our version is simply aimed to replace the `Wick isomorphism' of Ref. \cite{ablamowiczfauser}.
\begin{note}\label{n:wick}In Ref. \cite[Sec. Wick isomorphism]{ablamowiczfauser}  the authors consider two bilinear forms $g$ and $F$ related by $F=g+A,$ where $g$ is symmetric, while $A$ is alternating  and realized as $A(x,y) = \bi_F^g (x\wedge y)$,   where $x,y\in V,$ $F\in\bigwedge^2 (V),$ and, for $u,v\in V,$ $\bi_{u\w v}^g=\bi_u^g\circ\bi_v^g.$
The Clifford algebras $\cl(V,g)$ and $\cl(V,F)$ are realized as $\bigwedge(V)$ equipped with the algebra products defined by the general Chevalley's formula (\ref{eq:xFv}). The authors then state that `The Wick isomorphism is now given as $\cl(F,V)=\phi^{-1}(\cl(g,V))
= \mathrm{e}_\wedge^{-F}\wedge \cl(Q,V)\wedge \mathrm{e}_\wedge^F,$' where $\phi$ is not defined apart of this one formula. I have doubts about the validity of this construction. It is clear that $\phi^{-1}\wedge z\wedge \phi=z$ for all $z\in\bigwedge(V)$ if $\phi$ is an invertible element of $\bigwedge^+(V)$.
 Moreover, concerning the following `contraction', (as in the formula (3.1),(ii) of the paper), $z\mapsto (x\lrcorner_gF)\wedge z$ is an operator of degree $+1$ which cannot be equal to an operator of degree $-1$.
\end{note}
\begin{cor}
Assuming $\BK$ to be of characteristic zero, for any $u\in\cl(Q)$ and any alternating  form $A$ we have
\be \mathrm{e}^{\bi_{A^*}}(u)=\bl_A(u),\label{eq:blexp}\ee
where $\bl_A$ is given by  Eqs.  (\ref{eq:lb1}),(\ref{eq:lb2}),(\ref{eq:lb3}).
\end{cor}
\begin{proof}
When the field $\BK$ contains the field of rational numbers, the exponential $\exp(i_{A^*})\in \End(\cl(Q))$ is well defined. 
It follows then from Eqs. (\ref{eq:bif2}),(\ref{eq:ibfg}) that $[\,\bi_{f\wedge g},\bi_h]=0$ for all $f,g,h\in V^*,$ and therefore
$[\,\bi_{A^*},\bi_h]=0$ for all $A^*\in\bigwedge^2(V^*),h\in V^*$. In particular $[\,\bi_{A^*},\bi_x^A]=0$ for all $A^*\in\bigwedge^2(V^*),x\in V$. Therefore
\be
\begin{split} &\mathrm{e}^{\bi_{A^*}}\,\bar{e}_x \mathrm{e}^{-\bi_{A^*}}\\
&=\bar{e}_x+[\,\bi_{A^*},\bar{e}_x]+\frac12[\,\bi_{A^*},[\,\bi_{A^*},\bar{e}_x]]+
\frac{1}{3!}[\,\bi_{A^*},[\,\bi_{A^*},[\,\bi_{A^*},\bar{e}_x]]]+\cdots\\
&=\bar{e}_x+\bi^A_x+\frac12[\,\bi_{A^*},\bi_x^A]+
\frac{1}{3!}[\,\bi_{A^*},[\,\bi_{A^*},\bi_x^A]]+\cdots\\
&=\bar{e}_x+\bi_x^A,\label{eq:exe}\end{split}\ee
which we can rewrite as\footnote{Note that the exponential series in Eq. (\ref{eq:exe}) de facto terminate at $\mathbf{1}+\bi_{A^*}$, because all higher commutators vanish. Therefore, with such an understanding, our formulae also make sense in characteristic other than $0$. In fact, we could also express $\bl_F$ in terms of divided power exponential, as we have done it for $\lambda_F$ in Sec. \ref{sec:lexp}, rewriting those formulas using the appropriately changed notation.}
\be \mathrm{e}^{\bi_{A^*}}(xu)=x\left(\mathrm{e}^{\bi_{A^*}}(u)\right)+\bi_x^A\left(\mathrm{e}^{\bi_{A^*}}(u)\right).\label{eq:eias}\ee
Since $\bi_{A^*}\left(\mathbf{1}_{\cl(Q)}\right)=0,$ we have $\mathrm{e}^{\bi_{A^*}}\left(\mathbf{1}_{\cl(Q)}\right)=\mathbf{1}_{\cl(Q)}$, which together with Eq. (\ref{eq:eias}) shows that $\exp(\bi_{A^*})$ satisfies Eqs.  (\ref{eq:lb1}),(\ref{eq:lb2}), and therefore $\exp(\bi_{A^*})=\bl_A.$
\end{proof}
With $F=A$, $x\in V$, $v\in\cl(V,Q)$, we can now rewrite the second line in Eq. (\ref{eq:xFv}) as
\be \Tprod{x}{v}{A}=\mathrm{e}^{\bi_{A^*}}\left(x\,\mathrm{e}^{-\bi_{A^*}}(v)\right).\label{eq:xAv}\ee
\label{sec:aaf}\subsection{Symbol map and quantization map}
Let $V$ be a vector space over a field of characteristic $\neq 2$, $Q$ a quadratic form, and $\Phi$ the associated bilinear form. Let $F(x,y)=\frac{1}{2}\Phi(x,y),$ so that $Q(x)=F(x,x).$  We denote by $\sigma_Q$ the linear map
\be \sigma_Q:\cl(Q)\rightarrow \bigwedge(V)\nn\ee
defined as $\sigma_Q=\bl_F,$ and by $q_Q:\bigwedge(V)\rightarrow \cl(Q)$ the inverse of $\sigma_Q$:
\be q_Q=\sigma_Q^{-1}=\lambda_{-F}.\nn\ee

Following Ref. \cite[Ch. 2.2.5, p. 32]{meinrenken2013} we call $\sigma_Q$ the symbol map, and $q_Q$ the quantization map.
The following property of the symbol map follows as an immediate application of Corollary \ref{cor:sqo}.

\begin{cor}
If $V$ is vector space over a field of characteristic $\neq 2,$ if $x_1,\ldots ,x_n$ are in $V$, and if they are pairwise orthogonal, i.e. if $F(x_i,x_j)=0$ for $i\neq j$ then
\be \sigma_Q(x_1\cdots x_n)=x_1\wedge \cdots \wedge x_n.\label{eq:3c1}\ee
\label{cor:3c1}\end{cor}
We also notice the following important property of the quantization map:\footnote{Corollary \ref{cor:3c1} and Proposition \ref{prop:3c} are stated as an exercise in Ref. \cite[Exercise 3c, p. 154]{bourbaki1959}.}
\begin{proposition}
If $V$ is vector space over a field of characteristic $\neq 2,$ then for any sequence $y_1,\ldots,y_n\in V$ we have\be n!q_Q(y_1\wedge y_2\wedge \cdots \wedge y_n)=\sum_\sigma (-1)^\sigma\, y_{\sigma(1)}y_{\sigma(2)}\cdots y_{\sigma(n)}.\label{eq:nfac}\ee
\label{prop:3c}
\end{proposition}
\begin{proof}
Suppose $y_1,\ldots,y_n\in V$ are arbitrary. Then they span a finite-dimensional vector subspace of $V$, and we can choose there an orthogonal basis $e_\mu$, $\mu=1,2,\ldots,p$, $p\leq n$ (see e.g. Ref. \cite[Th. 1, p. 90]{bourbaki1959}). Since both sides of Eq. (\ref{eq:nfac}) are linear in $y_1,\ldots,y_n$, it is enough to verify the equality on basis vectors. Let us therefore assume that $y_1,\ldots,y_n$ are taken from the set of basis vectors $e_1,\ldots,e_p$. If any two $y_i,y_j$ are identical, then the left hand side of (\ref{eq:nfac}) vanishes because of the properties of the wedge product. But then also the right hand side  vanishes, because, owing to the orthogonality of the basis vectors, we can always arrange the product on the right hand side in such a way that the coinciding vectors are next to each other, and then there will be an identical term with the transposition of the coinciding vectors, with the minus sign - and the two terms will cancel. On the other hand, when all $y_i$ are different vectors from the basis $(e_\mu)$, and this can happen only when $p=n,$ the result follows by applying $q_Q$ to both sides of Eq. (\ref{eq:3c1}) and taking into account the fact that the wedge product changes sign when two of its factors are transposed.
\end{proof}
\section*{Acknowledgment}

The author would like to thank Rafa{\l} Ab{\l}amowicz for his interest and discussion, Robert Coquereaux for reading parts of the manuscript and numerous comments and suggestions, and my wife Laura for reading the manuscript and suggesting many corrections. Thanks are also due to Jacques Helmstetter for answering my questions and providing illuminating critical comments.


\begin{thebibliography}{99}
\bibitem{oziewicz97} Oziewicz, Z.: Clifford Algebra of
Multivectors, Proc. Int. Conf. on the Theory of the Electron
(Cuautitlan, Mexico, 1995). Eds. J.~Keller and Z.~Oziewicz,
Adv. in Appl. Clifford Alg. {\bf 7} (Suppl.), 467--486 (1997)

\bibitem{ablamowiczlounesto}
Ab{\l}amowicz, R., Lounesto, P.: On Clifford Algebras of a Bilinear
Form with an Antisymmetric Part, in ``Clifford Algebras with Numeric and
Symbolic Computations".
Eds. R. Ab{\l}amowicz, P. Lounesto, J.M. Parra, Birkh\"{a}user, Boston, (1996),
167--188.

\bibitem{fauser97}Fauser, B.: Clifford Geometric Parameterization
of Inequivalent Vacua, Mathematical Methods in the Applied Sciences {\bf 24}, p. 885-912 (2001).

\bibitem{ablamowiczfauser} Fauser, B., Ab{\l}amowicz, R.: On the Decomposition of Clifford Algebras of Arbitrary Bilinear Form, In "Clifford Algebras and their Applications in Mathematical Physics: Volume 1: Algebra and Physics", Birkh\"{a}user (2000) (see also arXiv:math/9911180, though with a different numeration of sections and equations.)

\bibitem{ablam2014}Ab{\l}amowicz, R., Gon\c{c}alves, I., da Rocha, R.: Bilinear Covariants and Spinor Fields Duality in Quantum Clifford Algebras, J. Math. Phys. {\bf 55} 103501 (2014)

\bibitem{chevalley1996}Chevalley, C.: The Algebraic Theory of Spinors and Clifford Algebras: Collected
  Works, Volume 2, Springer (1996), reprint of Chevalley, C.: The algebraic theory of spinors, Columbia University Press (1954)

\bibitem{bourbaki1959}Bourbaki, N.:Alg\`{e}bre, Chapitre 9, Formes sesquilin\'{e}ares et Formes Quadratiques, Hermann (1959)

\bibitem{meinrenken2013}Meinrenken, E.: Clifford Algebras and Lie Theory, Springer (2013)

\bibitem{Northcott2009}Northcott, D.~G.: Multilinear Algebra, Cambridge University Press, (2009)

\bibitem{bleecker2005}Bleecker, D.: Gauge Theory and Variational Principles, Dover (2005)

\bibitem{percacci1986}Percacci, R.: Geometry of Nonlinear Field Theories, World Scientific (1986)

\bibitem{obukhov1985}Ivanenko, D., Obukhov, Yu. N.: Gravitational Interaction of Fermion Antisymmetric Forms, Ann. der. Physik, {\bf 42}, 59-70 (1985)

\bibitem{obukhov1994}Obukhov, Yu. N., Solodukhin, S. N., Dirac Equation and the Ivanenko-Landau-K\"{a}hler Equation, Int. J. Theor. Phys. {\bf 33}, 225-245 (1994)

\bibitem{lounesto2001}Lounesto, P.: Clifford Algebras and Spinors, Secon Edition, Cambridge University Press (2001)

\bibitem{obukhov2017}Itin, Y., Hehl, F. W., Obukhov, Yu. N.: Premetric Equivalent of General Relativity, Phys. Rev. D {\bf 95}, o84020 (2017)

\bibitem{graf1978}Graf, W.: Differential Forms as Spinors, Ann. I.H.P {\bf 29}85-109 (1978)

\bibitem{Bourbaki1998}Bourbaki, N.: Algebra I: Chapters 1-3, Springer (1998)

\bibitem{greub1978}Greub, W.: Multilinear Algebra, 2nd ed, Springer (1978)

\bibitem{Bourbaki2006}Bourbaki, N.: Alg\`{e}bre -9, Springer-Verlag GmbH (2006)

\bibitem{dieu3}Dieudonn\'{e}, J.: Treatise on Analysis Volume III, Academic Press (1972)

\bibitem{Helmstetter2008}Helmstetter, J., Micali, A.: Quadratic Mappings and Clifford Algebras,
Birkh\"{u}ser (2008)

\bibitem{bourbaki_ru}Bourbaki, N.: Algebra. Moduli, kolca, formy, Nauka, Moskva (1966), in Russian

\bibitem{crumeyrolle}Crumeyrolle, A.: Diracian Theory of Fields - Spinorgravitation, Rep. Math. Phys. {\bf 31}, p. 29 (1992)

\bibitem{chevalley1956}Chevalley, C.: Fundamental Concepts of Algebra, Academic Press (1956)

\bibitem{bourbaki47}Bourbaki,N.: Alg\`{e}bre Chapitres 4 \`{a} 7, Springer (2007)

\bibitem{sulanke2008}Onishchik, A.~L., Sulanke, R.: Algebra und Geometrie: Moduln und Algebren,
VEB Deutscher Verlag der Wissenschaften (1988)

\bibitem{gilbert91}Gilbert, J. E, Murray M. A. M.: Clifford algebras and Dirac operators in harmonic analysis, Cambridge University Press (1971)

\bibitem{erdmann}Erdmann, K., Holm, T.: Algebras and Representation Theory, Springer (2018)

\bibitem{oziewicz86}Oziewicz, Z.: From Grassmann to Clifford, In Clifford Algebras and Their Applications in Mathematical
  Physics, Reidel (1986)
\end{thebibliography}
\end{document}